\documentclass[sigconf]{acmart}

\renewcommand\footnotetextcopyrightpermission[1]{} % removes footnote with conference information in first column
\usepackage{graphicx}
\usepackage{subcaption}
\usepackage{balance}  % for  \balance command ON LAST PAGE  (only there!)
\usepackage{amsmath,epsfig}
\usepackage{amssymb,amsthm,bbm}
\usepackage[noend]{algpseudocode}
\usepackage{algorithm}
\usepackage{epstopdf}
\usepackage{multi row}     % multi rows for tables
\usepackage[font={bf}, tableposition=bottom]{caption}     % captions on top for tables
\usepackage{microtype}    % compress text
\usepackage{units}     % nicer slanted fractions
\usepackage{xspace}     % fix space in macros
\usepackage{booktabs}     % nicer tables
\usepackage{siunitx}          % \num for decimal grouping}
\usepackage{xcolor}
\usepackage{tablefootnote}
\usepackage{float}
\usepackage{xfrac}
\usepackage{balance}  % for  \balance command ON LAST PAGE  (only there!)
\usepackage{type1cm}     % type1 computer modern font
\usepackage{todonotes}

\newcommand{\spara}[1]{\smallskip\noindent{\bf #1}}

\newtheorem{theorem}{Theorem}[section]

\newtheorem{lemma}[theorem]{Lemma}

\newtheorem{problem}[theorem]{Problem}

\newtheorem{claim}{Claim}[section]

\newcommand{\graph}{\ensuremath{G}\xspace}

\newcommand{\grapht}[1]{\ensuremath{\graph^{#1}}\xspace}
\newcommand{\vertices}{\ensuremath{V}}

\newcommand{\edges}{\ensuremath{E}\xspace}
\newcommand{\edge}{\ensuremath{e}\xspace}
\newcommand{\labels}{\ensuremath{L}\xspace}

\newcommand{\verticest}[1]{\ensuremath{\vertices^{#1}}\xspace}
\newcommand{\operation}{\textbf{o}\xspace}
\newcommand{\edgest}[1]{\ensuremath{\edges^{#1}}\xspace}

\newcommand{\tuple}[1]{\ensuremath{\langle {#1}\rangle}\xspace}
\newcommand{\subgraph}{\ensuremath{\graph_S}\xspace}
\newcommand{\subgraphone}{\ensuremath{\graph_S}\xspace}
\newcommand{\subgraphtwo}{\ensuremath{\graph_T}\xspace}
\newcommand{\subsetv}{\ensuremath{S}\xspace}
\newcommand{\subsetvone}{\ensuremath{S}\xspace}
\newcommand{\subsetvtwo}{\ensuremath{T}\xspace}

\newcommand{\subsete}[1]{\ensuremath{E(#1)}\xspace}

\newcommand{\isomorphismf}{\ensuremath{I}\xspace}
\newcommand{\numnodes}{\ensuremath{n}\xspace}
\newcommand{\numedges}{\ensuremath{m}\xspace}

\newcommand{\neighborst}[2]{\ensuremath{N_{#1}^{#2}}\xspace}
\newcommand{\degree}[1]{\ensuremath{d_{#1}}\xspace}
\newcommand{\degreet}[2]{\degree{#1}^{#2}\xspace}

\newcommand{\sample}{\ensuremath{\mathcal{S}}\xspace}

\newcommand{\bigO}{\ensuremath{\mathcal{O}}\xspace}

\newcommand{\indsubgraphs}[1]{\ensuremath{\mathcal{C}^{#1}}\xspace}
\newcommand{\indsubgraphsi}[2]{\ensuremath{\mathcal{C}^{#1}_{#2}}\xspace}
\newcommand{\numclasses}[1]{\ensuremath{T_{#1}}\xspace}
\newcommand{\freq}{\ensuremath{f}\xspace}
\newcommand{\freqindsubgraphs}[2]{\ensuremath{\mathcal{F}^{#1}_{#2}}\xspace}
\newcommand{\relfreq}[1]{\ensuremath{{p}_{#1}}\xspace}
\newcommand{\subgraphsample}{\ensuremath{\mathcal{S}}\xspace}
\newcommand{\approxindsubgraphs}[4]{\ensuremath{\mathcal{\tilde{F}}^{#1}_{#2}(#3,#4)}\xspace}
\newcommand{\idx}{\ensuremath{\mathcal{I}}\xspace}

\newcommand{\ec}{\ensuremath{\mathbf{EC}}\xspace} 
\newcommand{\er}{\ensuremath{\mathbf{ER}}\xspace} 
\newcommand{\sr}{\ensuremath{\mathbf{SR}}\xspace} 
\newcommand{\osr}{\ensuremath{\mathbf{OSR}}\xspace}

\newcommand{\nodelabelspace}{\ensuremath{L}\xspace}
\newcommand{\edgelabelspace}{\ensuremath{Q}\xspace}
\newcommand{\nodelabel}[1]{\ensuremath{\ell_{#1}}\xspace}
\newcommand{\edgelabel}[1]{\ensuremath{q_{#1}}\xspace}

\newcommand{\threshold}{\ensuremath{\tau}\xspace}

\algdef{SE}[DOWHILE]{Do}{doWhile}{\algorithmicdo}[1]{\algorithmicwhile\ #1}%

\newenvironment {squishlist}
{\begin{list}{$\bullet$}
  { \setlength{\itemsep}{1pt}
     \setlength{\parsep}{1pt}
     \setlength{\topsep}{1pt}
     \setlength{\partopsep}{1pt}
     \setlength{\leftmargin}{1.5em}
     \setlength{\labelwidth}{1em}
     \setlength{\labelsep}{0.5em} } }
{\end{list}}

\usepackage{environ}
\NewEnviron{myequation}{%
    \begin{equation}
    \scalebox{1}{$\BODY$}
    \end{equation}
    }

\errorcontextlines\maxdimen

\makeatletter
    \newcommand*{\algrule}[1][\algorithmicindent]{\makebox[#1][l]{\hspace*{.5em}\thealgruleextra\vrule height \thealgruleheight depth \thealgruledepth}}%
\newcommand*{\thealgruleextra}{}
\newcommand*{\thealgruleheight}{.75\baselineskip}
\newcommand*{\thealgruledepth}{.25\baselineskip}

\newcount\ALG@printindent@tempcnta
\def\ALG@printindent{%
    \ifnum \theALG@nested>0% is there anything to print
        \ifx\ALG@text\ALG@x@notext% is this an end group without any text?
        \else
            \unskip
            \addvspace{-1pt}% FUDGE to make the rules line up
            \ALG@printindent@tempcnta=1
            \loop
                \algrule[\csname ALG@ind@\the\ALG@printindent@tempcnta\endcsname]%
                \advance \ALG@printindent@tempcnta 1
            \ifnum \ALG@printindent@tempcnta<\numexpr\theALG@nested+1\relax% can't do <=, so add one to RHS and use < instead
            \repeat
        \fi
    \fi
    }%
\usepackage{etoolbox}
\patchcmd{\ALG@doentity}{\noindent\hskip\ALG@tlm}{\ALG@printindent}{}{\errmessage{failed to patch}}
\makeatother

\newbox\statebox
\newcommand{\myState}[1]{%
    \setbox\statebox=\vbox{#1}%
    \edef\thealgruleheight{\dimexpr \the\ht\statebox+1pt\relax}%
    \edef\thealgruledepth{\dimexpr \the\dp\statebox+1pt\relax}%
    \ifdim\thealgruleheight<.75\baselineskip
        \def\thealgruleheight{\dimexpr .75\baselineskip+1pt\relax}%
    \fi
    \ifdim\thealgruledepth<.25\baselineskip
        \def\thealgruledepth{\dimexpr .25\baselineskip+1pt\relax}%
    \fi
    \State #1%
    \def\thealgruleheight{\dimexpr .75\baselineskip+1pt\relax}%
    \def\thealgruledepth{\dimexpr .25\baselineskip+1pt\relax}%
}
\usepackage{tikz}

\DeclareGraphicsExtensions{.pdf,.png,.jpg,.eps}
\graphicspath{{./img/}}

\begin{document}

% ****************** TITLE ****************************************
\title{Mining Frequent Patterns in Evolving Graphs}

\author{Cigdem Aslay}
\affiliation{\institution{Aalto University}}
\email{cigdem.aslay@aalto.fi}

\author{Muhammad Anis Uddin Nasir}
\affiliation{\institution{King Digital Entertainment Ltd}}
\email{anis.nasir@king.se}

\author{Gianmarco De Francisci Morales}
\affiliation{\institution{ISI Foundation}}
\email{gdfm@acm.org}

\author{Aristides Gionis}
\affiliation{\institution{Aalto University}}
\email{aristides.gionis@aalto.fi}

%\author{
%Cigdem Aslay{\small $^{\sharp}$},
%Muhammad Anis Uddin Nasir{\small $^{\star}$}, \\
%Gianmarco De Francisci Morales{\small $^{\diamond}$}, 
%Aristides Gionis{\small $^{\sharp}$}\\
%}
%\affiliation{%
%  \institution{
%$^{\sharp}$Aalto University  \hspace{1mm}
%$^{\star}$King Digital Entertainment Ltd \hspace{1mm}
%$^{\diamond}$ISI Foundation \hspace{1mm}
%%$^{\sharp}$Aalto University \hspace{1mm}
%  }
%  {\small
%  cigdem.aslay@aalto.fi,
%  anis.nasir@king.se,
%  gdfm@acm.org,
%  aristides.gionis@aalto.fi
%  }
%}
\thanks{Part of the work was done while the first author was at ISI Foundation, the second author was at KTH, and the third author was at QCRI.}

\copyrightyear{2018} 
\acmYear{2018} 
\setcopyright{none}
\acmConference[CIKM '18]{The 27th ACM International Conference on Information and Knowledge Management}{October 22--26, 2018}{Torino, Italy}
\acmBooktitle{The 27th ACM International Conference on Information and Knowledge Management (CIKM '18), October 22--26, 2018, Torino, Italy}
\acmPrice{15.00}
\acmDOI{10.1145/3269206.3271772}
\acmISBN{978-1-4503-6014-2/18/10}

\begin{abstract}
Given a labeled graph, 
the frequent-subgraph mining (FSM) problem asks to find all the $k$-vertex subgraphs 
that appear with frequency greater than a given threshold. 
FSM has numerous applications ranging from biology to network science, 
as it provides a compact summary of the characteristics of the graph.
However, the task is challenging, 
even more so for evolving graphs due to the streaming nature of the input and 
the exponential time complexity of the problem.

In this paper, we initiate the study of the approximate FSM problem 
in both incremental and fully-dynamic streaming settings, 
where arbitrary edges can be added or removed from the graph. 
For each streaming setting, we propose algorithms that can extract 
a high-quality approximation of the frequent $k$-vertex subgraphs 
for a given threshold, at any given time instance, with high probability. 
In contrast to the existing state-of-the-art solutions that require iterating 
over the entire set of subgraphs for any update, 
our algorithms operate by maintaining a uniform sample of $k$-vertex subgraphs 
% at any time instance 
with optimized neighborhood-exploration procedures local to the updates. 
We provide theoretical analysis of the proposed algorithms and 
empirically demonstrate that the proposed algorithms generate high-quality results compared to baselines.
\end{abstract}

\maketitle
% !TEX root = ../main.tex
\section{Introduction}
\label{sec:intro}

\enlargethispage{\baselineskip}

% FSM definition and relevance
Frequent-subgraph mining (FSM) is a fundamental graph-mining task with applications in various disciplines, including bioinformatics, security, and social sciences.
The goal of FSM is to find subgraph patterns of interest that are frequent in a given graph.
Such subgraphs might be indicative of an important protein interaction, a possible intrusion, or a common social norm.
FSM also finds applications in graph classification and indexing.

% Challenges
Existing algorithms for subgraph mining are not scalable to large graphs that arise, for instance, in social domains.
In addition, these graphs are usually produced as a result of a dynamic process, hence are subject to continuous changes.
For example, in social networks new edges are added as a result of the interactions of their users, and the graph structure is in continuous flux.
Whenever the graph changes, i.e., by adding or removing an edge, a large number of new subgraphs can be created, and existing subgraphs can be modified or destroyed.
Keeping track of all the possible changes in the graph is subject to combinatorial explosion, thus, is highly challenging.

% Existing attempts
In this paper we address the problem of mining frequent subgraphs in an evolving graph, 
which is represented as a stream of edge updates --- additions or deletions.
Only a few existing works consider a similar setting~\citep{bifet2011mining,ray2014frequent,abdelhamid2017incremental}.
\citet{bifet2011mining} deal with a transactional setting where the input is a stream of small graphs.
% rather than a single graph that evolves over time.
Their setting is similar to the one considered by frequent-itemset mining, so many of the existing results can be reused.
Conversely, in our case, there is a single graph that is continuously evolving.
\citet{ray2014frequent} consider a scenario similar to the one we study in this paper.
They consider a single graph with continuous updates, although they only allow incremental ones (edge addition) rather than the fully-dynamic ones we consider (edge addition and deletion).
Moreover, their approach is a simple heuristic that does not provide any correctness guarantee.
Our approach, instead, is able to provably find the frequent subgraphs in a fully dynamic graph stream.
\citet{abdelhamid2017incremental} tackle a problem setting similar to ours, with a single fully-dynamic evolving graph.
They propose an exact algorithm which tracks patterns which are at the ``fringe'' of the frequency threshold, and borrows heavily from existing literature on incremental pattern mining.
As such, they need to use a specialized notion of frequency for graphs (minimum image support).
Instead, our algorithm provides an approximate solution which uses the standard notion of induced subgraph isomorphism for frequency.

% Our contribution
This paper is the first to propose an approximation algorithm for the frequent-subgraph mining problem on a fully-dynamic evolving graph.
We propose a principled sampling scheme for subgraphs and provide theoretical justifications for its accuracy.
Differently from previous work on sampling from graph streams, 
our method relies on sampling \emph{subgraphs} rather than edges.
This choice enables sampling any kind of subgraph of the same size with equal probability, and thus simplifies dramatically the design of the frequency estimators.
We maintain a uniform sample of subgraphs via reservoir sampling, 
which in turn allows us to estimate the frequency of different patterns.
To handle deletions in the stream, we employ an adapted version of \emph{random pairing}~\cite{gemulla2006dip}.
Finally, to increase the efficiency of our sampling procedure during the exploration of the local neighborhood of updated edges, we employ an adaptation of the ``skip optimization,'' proposed by~\citet{vitter1985random} for reservoir sampling and by~\citet{gemulla2008maintaining} for random pairing. %Finally, to make our sampling procedure more efficient, we adapt \citeauthor{vitter1985random}'s idea of ``skip optimization'' for reservoir sampling~\cite{vitter1985random} to the exploration of the neighborhood of the edge update. %exploration of the neighborhood local to the edge updates

\enlargethispage{\baselineskip}

% Our results
Concretely, our main contributions are the following:
\begin{squishlist}
\item We are the first to propose an approximation algorithm for the frequent-subgraph mining problem for evolving graph.
\item We propose a new subgraph-based sampling scheme.
\item We show how to use random pairing to handle deletions.
\item We describe how to implement neighborhood exploration efficiently via ``skip optimization.''
\item We provide theoretical analysis and guarantees on the accuracy of the algorithm.
\end{squishlist}

% !TEX root =  ../main.tex
%% PITCH
\section{Problem definition}
\label{sec:prel}

We consider graphs with vertex and edge \emph{labels}.
We model dynamic graphs as a sequence of \emph{edge} \emph{additions} and \emph{deletions}. 

We assume monitoring a graph that changes over time. 
For any time $t\geq 0$, we let $\grapht{t}=(\verticest{t}, \edgest{t})$ 
be the graph that has been observed up to and including time $t$, 
where $\verticest{t}$ represents the set of vertices and $\edgest{t}$ represents the set of edges.
We assume that vertices and edges have labels, 
and we write $\nodelabelspace$ and $\edgelabelspace$ 
for the sets of labels of vertices and edges, respectively.
For each vertex $v\in\verticest{t}$ we denote its label by $\nodelabel{v} \in \nodelabelspace$, 
and similarly, 
for each edge $e=(u,v)\in\edgest{t}$ we denote its label by $\edgelabel{e} \in \edgelabelspace$. 
% Each vertex in $\verticest{t}$ is a key-value pair $\tuple{{v,\nodelabel{v}}}$, where $v$ is the vertex and $\nodelabel{v} \in \nodelabelspace$ is the label of the vertex $v$. Similarly, each edge in $\edgest{t}$ is of the form $\tuple{{u,v, \edgelabel{e}}}$, where $u$ and $v$ are two vertices and $\edgelabel{e} \in \edgelabelspace$ is the edge label.
Initially, at time $t=0$, we have $\verticest{0}=\edgest{0}=\varnothing$.
For any $t \ge 0$, at time $t+1$ we receive an update tuple 
$\tuple{\operation, \edge, \edgelabel{}}$ from a stream, 
where 
$\operation \in \{+,-\}$ represents an update operation, addition or deletion, 
$\edge=(u,v)$ is a pair of vertices, and
$\edgelabel{}\in\edgelabelspace$ is an edge label.
The graph $\grapht{t+1}=(\verticest{t+1}, \edgest{t+1})$ is obtained by adding a new edge 
or deleting an existing edge as follows:
\[
\edgest{t+1} =
  \begin{cases}
    \edgest{t} \cup \{\edge\} \mbox{ and } \edgelabel{\edge}=\edgelabel{} 
     & \quad \text{if } \operation = + \\
    \edgest{t} \setminus \{\edge\} & \quad \text{if } \operation = -  \quad.
  \end{cases}
\]
Additions and deletions of vertices are treated similarly.
Furthermore, we assume that when adding an edge $\edge=(u,v)$, 
the vertices $u$ and $v$ are added in the graph too, 
if they are not present at time~$t$.
Similarly, when deleting a vertex, 
we assume that all incident edges are deleted too, prior to the vertex deletion. 
Our model deals with the fully dynamic stream of edges, which is different from the stream of graphs \cite{wackersreuther2010frequent}.
For simplicity of exposition, in the rest of the paper 
we discuss only edge additions and deletions ---
vertex operations can be handled rather easily.

%Consider an undirected labeled graph $\dgraph{\vertices}{\edges}{\labels}{\labelf}$.
%$\vertices$ are the set of vertices, $\edges$ are the set of edges, $\labels$ are the set of labels.
%Further, $\labelf$ is a labeling function that maps the set of vertices and edges to set of labels, i.e., $\labelf \colon 
%\vertices \cup \edges \to \labels$.
%Each edge $\edge \in \edges$ can be denoted by a pair of vertices $(\vertex_i,\vertex_j)$ where $\vertex_i,\vertex_j \in \vertices$.
We use $\numnodes_t=\vert \verticest{t} \vert$ and $\numedges_t=\vert \edgest{t} \vert$ 
to refer to the number of vertices and edges, respectively, at time $t$.
%A graph without a self-loop or multi-edge is a simple graph.
In this work, we considered simple, connected, and undirected graphs.
The \emph{neighborhood} of a vertex $u\in \verticest{t}$ at time~$t$ is defined as 
$\neighborst{u}{t} =\{v \mid (u,v) \in \edgest{t}\}$, 
and its \emph{degree} as $\degreet{u}{t} =| \neighborst{u}{t} |$.
Similarly, the $h$-hop neighborhood of $u$ at time $t$ is denoted as $\neighborst{u,h}{t}$, and 
indicates the set of the vertices that can be reached from $u$ in $h$ steps by following the edges \edgest{t}.
To simplify the notation, we omit to specify the dependency on $t$ when it is obvious from the context.

For any graph $\graph=(\vertices,\edges)$ and a subset of vertices $\subsetv \subseteq \vertices$, 
we say that $\subgraph=(\subsetv,\subsete{\subsetv})$ is an \emph{induced} subgraph of $G$ if 
%$\subsetv \subseteq V$, $\subsete{\subsetv} \subseteq E$ and $\{e=(u,v) : u,v \in \subsetv, e \in E, e \not\in \subsete{\subsetv}\} = \varnothing$.
for all pairs of vertices $u,v \in \subsetv$ 
it is $(u,v) \in \subsete{\subsetv}$ if and only if $(u,v) \in \edges$.
We define $\indsubgraphs{k}$ to be the set of all induced subgraphs with $k$ vertices in~$\graph$. 
All subgraphs considered in this paper are induced subgraphs, 
unless stated otherwise.

% \todo[inline]{Shall we clearly state that we are working only with induced subgraphs in this paper hence we will simply refer to ``induced subgraph" as  ``subgraph" and explicitly state otherwise if necessary? Otherwise it takes a lot of space to write induced each time. --Cigdem}

%We say that two subgraphs of $G$, $\subgraph \in \indsubgraphs{k}$ and $\ensuremath{G_{s'}} \in C^k$ are isomorphic, denoted by $\subgraph \cong \ensuremath{\subgraphtwo}$, if there exists a bijection $\isomorphismf : \subsetv \mapsto \ensuremath{S'}$ such that $(u,v) \in E(S)$ implies $(\isomorphismf(u),\isomorphismf(v)) \in E(S')$ and $\isomorphismf$ preserves the vertex and edge labels, i.e., $L_{S}(u) = L_{S'}(\isomorphismf(u))$ and $L_{S}(u,v) = L_{S'}(\isomorphismf(u), \isomorphismf(v))$, $\forall (u,v) \in E(S)$. We partition $\indsubgraphs{k}$ into $\numclasses{k}$ equivalence classes $\indsubgraphsi{k}{i}, \cdots, \indsubgraphsi{k}{T_{k}}$, where subgraphs within $\indsubgraphsi{k}{i}$ are isomorphic and any pair of subgraphs from $\indsubgraphsi{k}{i}$ and $\indsubgraphsi{k}{j}$, $i\neq j$ are not isomorphic. 

We say that two subgraphs of $\graph$, 
denoted by $\subgraphone \in \indsubgraphs{k}$ and $\subgraphtwo \in \indsubgraphs{k}$ 
are isomorphic
if there exists a bijection $\isomorphismf : \subsetvone \rightarrow \subsetvtwo$ 
such that $(u,v) \in \edges(\subsetvone)$ if and only if $(\isomorphismf(u),\isomorphismf(v)) \in E(\subsetvtwo)$
and the mapping $\isomorphismf$ preserves the vertex and edge labels, 
i.e., $\nodelabel{u} = \nodelabel{\isomorphismf(u)}$ and 
$\edgelabel{(u,v)} = \edgelabel{(\isomorphismf(u), \isomorphismf(v))}$, 
for all $u\in\subsetvone$ and for all $(u,v) \in \edges(\subsetvone)$.
We write $\subgraphone \simeq \subgraphtwo$ to denote that $\subgraphone$ and $\subgraphtwo$ are isomorphic.

The isomorphism relation $\simeq$ partitions the set of subgraphs $\indsubgraphs{k}$
into $\numclasses{k}$ equivalence classes,\footnote{Notice that the value of $\numclasses{k}$ is simply determined by $k$, $\lvert L \rvert$, and $\lvert Q \rvert$.} 
denoted by $\indsubgraphsi{k}{1}, \cdots, \indsubgraphsi{k}{T_{k}}$.
% We partition $\indsubgraphs{k}$ into $\numclasses{k}$ equivalence classes 
% $\indsubgraphsi{k}{1}, \cdots, \indsubgraphsi{k}{T_{k}}$, 
% where subgraphs within $\indsubgraphsi{k}{i}$ are isomorphic, 
% and any pair of subgraphs from $\indsubgraphsi{k}{i}$ and $\indsubgraphsi{k}{j}$, 
% with $i\neq j$ are not isomorphic.
Each equivalence class $\indsubgraphsi{k}{i}$ is called a \emph{subgraph pattern}.

We define the support set $\sigma(\subgraph)$ of any $k$-vertex subgraph 
$\subgraph \in \indsubgraphs{k}$ as the number of $k$-vertex subgraphs of $G$ 
that are isomorphic to $\subgraph$, 
i.e., $\sigma(\subgraph) = \vert\indsubgraphsi{k}{i}\vert$, where $\subgraph \in \indsubgraphsi{k}{i}$. 
We then define the frequency $\freq(\subgraph)$ of a subgraph \subgraph 
as the fraction of $k$-vertex subgraphs of $\graph$ that are isomorphic to $\subgraph$, 
i.e., $\freq(\subgraph) = \sigma(\subgraph)/\vert \indsubgraphs{k} \vert$. 

Next we define the problem of mining frequent $k$-vertex subgraphs. 
Given a graph $\graph=(\vertices,\edges,\nodelabelspace, \edgelabelspace)$ 
and a frequency threshold $\tau \in (0,1]$, 
the set $\freqindsubgraphs{k}{\tau} \subseteq \indsubgraphs{k}$ 
of frequent $k$-vertex subgraphs of $\graph$ with respect to $\tau$ 
is the collection of all $k$-vertex subgraphs with frequency at least $\tau$, 
that is 
\[
\freqindsubgraphs{k}{\tau} = \{\subgraph \mid \subgraph \in \indsubgraphs{k} 
\mbox{ and } \freq(\subgraph) \ge \tau\}.
\]

\begin{problem}\label{pr:exactFrequent}
Given a graph $\graph=(\vertices,\edges,\nodelabelspace, \edgelabelspace)$, an integer $k$, 
and a frequency threshold~$\tau$, 
find the collection $\freqindsubgraphs{k}{\tau}$
% = \{\subgraph: \subgraph \in \indsubgraphs{k}, \freq(\subgraph) \ge \tau\}$ 
of frequent $k$-vertex subgraphs of $\graph$.
\end{problem}

Let $\relfreq{i} = \vert \indsubgraphsi{k}{i} \vert / \vert \indsubgraphs{k} \vert$ 
denote the frequency of isomorphism class $i$, with $i=1,\ldots,\numclasses{k}$. 
The problem of finding the frequent $k$-vertex subgraphs 
% for a given threshold $\tau$ 
requires finding all isomorphism classes 
$\indsubgraphsi{k}{i}$ 
% with $i=1,\ldots,\numclasses{k}$, such that 
with $\relfreq{i} \ge \tau$. 
Hence, we equivalently have
\[
\freqindsubgraphs{k}{\tau} = \bigcup_{i \in [1,\numclasses{k}]} 
\{\subgraph \mid \subgraph \in \indsubgraphsi{k}{i} \mbox{ and } \relfreq{i} \ge \tau\}.
\]

In this paper, our aim is to find an approximation to the collection $\freqindsubgraphs{k}{\tau}$ 
by efficiently estimating $\relfreq{i}$, % for all $i=1,\ldots,\numclasses{k}$, 
from a uniform sample $\subgraphsample$ of $\indsubgraphs{k}$. 
We say that a subset $\subgraphsample \subseteq \indsubgraphs{k}$, 
with $\vert \subgraphsample \vert = M$, 
is a uniform sample of size $M$ from $\indsubgraphs{k}$ 
if the probability of sampling $\subgraphsample$ 
is equal to the probability of sampling any $\subgraphsample' \subseteq \indsubgraphs{k}$ 
with $\vert \subgraphsample' \vert = M$, i.e., 
all samples of the same size are equally likely to be produced. 

Formally, we want to find an $(\epsilon, \delta)$-approximation to $\freqindsubgraphs{k}{\tau}$, 
denoted by $\approxindsubgraphs{k}{\tau}{\epsilon}{\delta}$ such that 
\[
\approxindsubgraphs{k}{\tau}{\epsilon}{\delta} = 
\bigcup_{i \in [1,\numclasses{k}]} 
\{\subgraph \mid \subgraph \in \indsubgraphsi{k}{i} \cap \subgraphsample, 
|\hat{\relfreq{i}}  - \relfreq{i}| \le \epsilon / 2, 
\relfreq{i} \ge \tau \},
\]
where $\hat{p}_i$ is the estimation of $\relfreq{i}$ 
such that % for each $i=1,\ldots,\numclasses{k}$, it is 
$|\hat{\relfreq{i}}  - \relfreq{i}| \le \epsilon / 2$ holds with probability at least $1 - \delta$.
In practice, the collection $\approxindsubgraphs{k}{\tau}{\epsilon}{\delta}$
of approximate frequent patterns
is computed from a sample $\subgraphsample \subseteq \indsubgraphs{k}$. 

The problem of approximate frequent subgraph mining 
can now be formulated as follows.

\begin{problem}\label{pr:approxFrequent}
Given a graph $\graph=(\vertices,\edges,\nodelabelspace, \edgelabelspace)$, a frequency threshold~$\tau$,
a small integer $k$, and constants $0< \epsilon, \delta < 1$, 
find the collection $\approxindsubgraphs{k}{\tau}{\epsilon}{\delta}$ 
that is an $(\epsilon, \delta)$-approximation to $\freqindsubgraphs{k}{\tau}$. 
\end{problem}

We focus on the dynamic case with vertex and edge additions and insertions.
As discussed above, at each time $t$ we consider the $\grapht{t}=(\verticest{t}, \edgest{t})$ that results from all vertex and edge operations.
Our goal is to maintain the approximate collection of frequent subgraphs $\approxindsubgraphs{k}{\tau}{\epsilon}{\delta}$ at each time $t$ without having to recompute it from scratch after each addition or deletion.

In the following problem definition we assume that vertex/edge labels are specified when a vertex/edge is added in the graph stream and they do not change afterwards. 
We make this assumption without loss of generality, as a vertex/edge label change can be simulated by a vertex/edge deletion followed by an addition of the same vertex/edge with different label.

\begin{problem}\label{pr:approxFrequent}
Given an evolving graph $\grapht{t}=(\verticest{t}, \edgest{t},\nodelabelspace, \edgelabelspace)$, 
a frequency threshold~$\tau$, a small integer $k$, and 
constants $0< \epsilon, \delta < 1$, 
maintain an approximate collection of frequent subgraphs 
$\approxindsubgraphs{k}{\tau}{\epsilon}{\delta}$ 
at each time~$t$. 
\end{problem}

\section{Algorithms}
\label{sec:algo}

This section describes the proposed algorithms, which are based on subgraph sampling.
We present two algorithms, both of which are based on two components: a reservoir of samples and an exploration procedure.
The goal of the reservoir is to capture the changes to already sampled connected $k$-subgraphs.\footnote{Hereafter, we simply refer to a $k$-vertex induced subgraph as $k$-subgraph.} The goal of the exploration procedure is to include newly (dis)connected $k$-subgraphs into the sample.
This separation of concerns allows the algorithm to minimize the amount of work per sample,
e.g., by avoiding computation of expensive minimum DFS codes for the corresponding patterns~\citep{yan2002gspan}.

The base algorithm requires to enumerate, at each time $t$, every newly (dis)connected $k$-subgraph at least once, by performing a \emph{neighborhood exploration} of the updated edge.
We show how to improve this algorithm by avoiding to materialize all the subgraphs via a \emph{skip optimization}.
This optimization enables picking subgraphs into the sample without having to list them all.
We also propose an additional heuristic to speed up the neighborhood exploration.
We provide an efficient implementation for the case $k = 3$, and describe how it generalizes to values $k > 3$ (although not as efficiently).

\subsection{Incremental streams}

We begin by describing our algorithm for maintaining a uniform sample 
\sample of fixed-size $M$ of $k$-subgraphs of $\grapht{t}$ for incremental streams (only edge addition).
The algorithm relies on \emph{reservoir sampling} \cite{vitter1985random} 
to ensure the uniformity of the sample \sample.

The addition of an edge $(u,v) \not\in \edgest{t-1}$ at time $t$ affects  only the subgraphs in the local neighborhoods up to $\neighborst{u,h}{t-1}$ and $\neighborst{v,j}{t-1}$, 
where $h+j=k-2$, i.e., all the connected $k$-subgraphs that contain $u$, $v$, 
and $h+j$ additional nodes from their neighborhoods, for all admissible values of $h,j \ge 0$.
Therefore, a uniform sample $\sample$ of subgraphs can be maintained by iterating through the subgraphs in the neighborhood of the newly inserted edge.
In particular, consider the addition of an edge $(u,v)$ at time $t$. 
Let $H \subseteq \neighborst{u,h}{t-1} \cup \neighborst{v,j}{t-1}$
be a subset of vertices, for some $h$ and $j$, such that 
$h+j=k-2$, 
$\{u,v\} \in H$, 
$|H| = k$.
There are two possible cases: ($i$) if $H$ is connected in \grapht{t-1}, a modified subgraph $H' = H \cup \{(u,v)\}$ is formed in \grapht{t}; ($ii$) if $H$ is not connected in \grapht{t-1}, and  $H' = H \cup \{(u,v)\}$ is connected in \grapht{t},  $H'$ is a newly formed connected $k$-subgraph in~\grapht{t}.

\spara{Example for $k=3$.}
Assume an edge $(u,v)$ arrives at time $t$.
For case ($i$) to hold, 
there should be some $w \in \neighborst{u}{t-1} \cap \neighborst{v}{t-1}$ for which the edge $(u,v)$ closes the wedge $\wedge = \{(u,w), (w,v)\}$ at \grapht{t-1}, 
forming a new triangle $\Delta = \wedge \cup \{(u,v)\}$ in $\grapht{t}$.
For case ($ii$) to hold, there must be some $w \in \neighborst{u}{t-1}$ 
(or $w \in \neighborst{v}{t-1}$), 
for which a new wedge $\{(u,v), (u,w)\}$ 
(respectively,  $\{(u,v), (v,w)\}$) is formed in~$\grapht{t}$.
\hfill$\qed$
\smallskip

When a modified subgraph $H'$ is formed in $G^{t}$, if the previously connected subgraph $H = H' \setminus \{(u,v)\}$ is present in $\sample$, 
we update the sample by substituting $H'$ with $H$. Otherwise, we ignore the modified subgraph. Given that the elements in the sample are induced connected subgraphs, this operation is equivalent to maintaining the sample up-to-date.

Conversely, when a new connected $k$-subgraph $H'$ is formed in 
$\grapht{t}$, we can be sure that it appears at time $t$ for the first time.
Therefore, we use the standard reservoir sampling algorithm as follows:
If $|\sample| < M$, we directly add the new subgraph $H'$ to the sample $\sample$.
Otherwise, if $|\sample| = M$, 
we remove a randomly selected subgraph in $\sample$ 
and insert the new one $H'$ with probability $M / N$, 
where $M$ is the upper bound on the sample size % (i.e., reservoir size) 
and $N$ is the total number of (valid) $k$-subgraphs \emph{encountered} since $t=0$.%
\footnote{Note that the addition of an edge $(u,v)$ translates to partially-dynamic 
$k$-subgraph streams in which the $k$-subgraphs are subject to addition and deletion operations, 
while $k$-cliques are subject to addition-only operations.
Thus, we can impose, without loss of generality, an order of operation 
during the exploration of the neighborhood of the inserted edge.} 
The modification of existing subgraphs in $\grapht{t}$ (i.e., case ($i$)) does not affect $N$, since, by definition, they replace the previous subgraphs which were already present in $\grapht{t-1}$. Therefore, the only increase in the number $N$ of subgraphs occurs in the case of new connected $k$-subgraph formations in $\grapht{t}$ (i.e., case ($ii$)).

Algorithm~\ref{alg:insertionOnly} shows the pseudocode for incremental streams. Next, we show that the sample $\sample$ maintained by Algorithm~\ref{alg:insertionOnly} is uniform at any given time $t$. 
%Next, we show that it ensures uniformity of the sample $\sample$.

\begin{claim}\label{claim:edgeInsertionOnly}
Algorithm~\ref{alg:insertionOnly} ensures the uniformity of the sample $\sample$ at any time $t$. 
\end{claim}
\begin{proof}
To show that $\sample$ is uniform, we need to consider two cases:
($i$) the inserted edge modifies an existing $k$-subgraph;
($ii$) the inserted edge forms a newly connected $k$-subgraph.

For the case of new subgraph formation, 
the uniformity property directly holds as it leverages the standard reservoir sampling algorithm.
Now, we show that the uniformity property holds when a subgraph is modified.

Assume the edge $(u,v) \not\in E^{t-1}$ is inserted at time $t$ and let $H$ denote the invalidated subgraph that is modified as $H' = H \cup \{(u,v)\}$ at time $t$. %for some $ \in N^{t-1}_u \cap N^{t-1}_v$.
Let $\sample'$ denote the sample after the invalidation of $H$ and the formation of $H'$.
For the sample to be truly uniform, the probability that $H' \in \sample'$ should be equal to $M/N$, conditioned on $\sample = M < N$ (conditioning on $\sample = N < M$ is trivial since every $k$-subgraph of $G^{t}$ would then be deterministically included in $\sample'$).
Now, given that $Pr\left(H \in \sample\right) =  M/N$, we have that
\begin{align*}
Pr\left(H' \in \sample' \right) &= Pr\left(H \in \sample, H' \in \sample'\right)  + Pr\left(H \not\in \sample, H' \in \sample'\right) \\
&= \dfrac{M}{N} \cdot 1 + \left(1 - \dfrac{M}{N}\right) \cdot 0 = \dfrac{M}{N}, \\
\end{align*} 
hence uniformity is preserved.
\end{proof}

%\begin{algorithm}[tb]
%\footnotesize
%\caption{Algorithm for Incremental Stream}
%\label{alg:insertionOnly}
%\begin{algorithmic}[1]
%\State $N \gets 0$, $\sample \gets \emptyset$
%\State $M \gets \log{\left({\numclasses{k}}/{\delta}\right)} \cdot {(4 + \epsilon)}/{\epsilon^2}$ 
%\Procedure{addEdge}{$t,(u,v)$}
%%	\State $N^{t-1}_{uv} \gets N^{t-1}_u \cap N^{t-1}_v$
%		\For{\texttt{$w \in N^{t-1}_{u} \cup N^{t-1}_{v}$}}
%			\If{$w \in N^{t-1}_u \cap N^{t-1}_v$} 
%				\State $\wedge \gets \{(u,w), (w,v)\}$ 
%				\If {$\wedge \in \sample$}
%					\State $\Delta \gets \wedge \cup \{(u,v)\}  $ 
%					\State \Call{Replace}{$\sample, \wedge, \Delta$} \Comment{replace H with H''}
%				\EndIf
%			\ElsIf{$w \in N^{t-1}_{u}$}
%				\State $\wedge \gets \{(u,w), (u,v)\}  $
%				\State \Call{ReservoirSampling}{$\wedge, \sample, M, N$} \Comment{Algorithm~\ref{alg:reservoirSampling}}
%			\ElsIf{$w \in N^{t-1}_{v}$}
%				\State $\wedge \gets \{(v,w), (u,v)\} $
%				\State \Call{ReservoirSampling}{$\wedge, \sample, M, N$} \Comment{Algorithm~\ref{alg:reservoirSampling}}
%			\EndIf
%		\EndFor
%\EndProcedure
%\Procedure{Replace}{$\sample, \ensuremath{G_{s'}}, \subgraph$}
%	\State $\sample \gets \sample \setminus \{\ensuremath{G_{s'}}\}$
%	\State $\sample \gets \sample \cup \{\subgraph\}$
%\EndProcedure
%\end{algorithmic}
%\end{algorithm}

\begin{algorithm}[tb]
\footnotesize
\caption{Algorithm for Incremental Stream}
\label{alg:insertionOnly}
\begin{algorithmic}[1]
\State $N \gets 0$, $\sample \gets \emptyset$
\State $M \gets \log{\left({\numclasses{k}}/{\delta}\right)} \cdot {(4 + \epsilon)}/{\epsilon^2}$ 
\Procedure{addEdge}{$t,(u,v)$}
	\For{$h \in [0, k-2]$}
		\State $j \gets k - 2 - h$
		\For{$H \subseteq \neighborst{u,h}{t-1} \cup \neighborst{v,j}{t-1}$}
			\If{$H$ is connected in $\grapht{t-1}$}
				\If {$H \in \sample$}
					\State $H' \gets H \cup \{(u,v)\}$ 
					\State \Call{Replace}{$\sample, H, H'$} \Comment{replace H with H'}
				\EndIf
			\Else
				\State $H' \gets H \cup \{(u,v)\}$  \Comment{H' is connected in $\grapht{t}$}
				\State \Call{ReservoirSampling}{$H', \sample, M, N$} %\Comment{Algorithm~\ref{alg:reservoirSampling}}
			\EndIf
		\EndFor
	\EndFor
\EndProcedure
\Procedure{Replace}{$\sample, \ensuremath{G_{R}}, \subgraph$}
	\State $\sample \gets \sample \setminus \{\ensuremath{G_{R}}\}$
	\State $\sample \gets \sample \cup \{\subgraph\}$
\EndProcedure
\end{algorithmic}
\end{algorithm}

\begin{algorithm}[tb]
\footnotesize
\caption{Fully-Dynamic-Edge Stream}
\label{alg:fullyDynamicExact}
\begin{algorithmic}[1]
\State $N \gets 0$, $\sample \gets \emptyset$, $c_1 \gets 0$, $c_2 \gets 0$
\State $M \gets \log{\left({\numclasses{k}}/{\delta}\right)} \cdot {(4 + \epsilon)}/{\epsilon^2}$ 
\Procedure{addEdge}{$t,(u,v)$}
%	\State $N^{t-1}_{uv} \gets N^{t-1}_u \cap N^{t-1}_v$
	%%%%%%%%%%%%%%%%%%%%%%%%%%%%%%%
	\For{$h \in [0,k-2]$}
		\State $j \gets k - 2 - h$
		\For{$H \subseteq \neighborst{u,h}{t-1} \cup \neighborst{v,j}{t-1}$}
			\If{$H$ is connected in $\grapht{t-1}$}
				\If {$H \in \sample$}
					\State $H' \gets H \cup \{(u,v)\}$ 
					\State \Call{Replace}{$\sample, H, H'$} \Comment{replace H with H'}
				\EndIf
			\Else
				\State $H' \gets H \cup \{(u,v)\}$  \Comment{H' is newly connected in $\grapht{t}$}
			\State \Call{RandomPairing}{$H', \sample, M$}
			\EndIf
		\EndFor
	\EndFor
\EndProcedure
\Procedure{deleteEdge}{$t,(u,v)$}
%\State $N^{t-1}_{uv} \gets N^{t-1}_u \cap N^{t-1}_v$
	\For{$h \in [0,k-2]$}
		\State $j \gets k - 2 - h$
		\For{$H \subseteq \neighborst{u,h}{t-1} \cup \neighborst{v,j}{t-1}$}
			\If{$H$ is still connected in $\grapht{t}$}
				\If {$H \in \sample$}
					\State $H' \gets H \setminus \{(u,v)\}$ 
					\State \Call{Replace}{$\sample, H, H'$} \Comment{replace H with H'}
				\EndIf
			\Else
				\If {$H \in \sample$}
					\State $\sample \gets \sample \setminus H$
					\State $c_1 \gets c_1 + 1$
				\Else 
					\State $c_2 \gets c_2 + 1$
				\EndIf
				\State $N \gets N-1$
			\EndIf
		\EndFor
	\EndFor

%		\Else
%			\If{$w \in N^{t-1}_{u}$}
%				\State $\wedge \gets \{(u,w), (u,v)\}  $
%%			\EndIf
%			\Else
%			%\If{$w \in N^{t-1}_{v}$}
%				\State $\wedge \gets \{(v,w), (u,v)\}  $
%			\EndIf
%			\If {$\wedge \in \sample$} 
%				\State $\sample \gets \sample \setminus \wedge$
%				\State $c_1 \gets c_1 + 1$
%			\Else 
%				\State $c_2 \gets c_2 + 1$
%			\EndIf
%			\State $N \gets N-1$
%		\EndIf
%	\EndFor	
\EndProcedure
\Procedure{RandomPairing}{$\subgraph,\sample, M$}
	\If{$c_1 + c_2 = 0$}
		\State \Call{ReservoirSampling}{$\subgraph, \sample, M, N$} %\Comment{Algorithm~\ref{alg:reservoirSampling}}
	\Else
		\If {\Call{uniform} $ < \frac{c_1}{c_1+c_2}$}
			\State $\sample \gets \sample \cup \subgraph$
			\State $c_1 \gets c_1 - 1$		
			%\EndIf
		\Else 
			\State $c_2 \gets c_2 -1$
		\EndIf				
	\EndIf
\EndProcedure
%\Procedure{Replace}{$\sample, \ensuremath{G_{R}}, \subgraph$}
%	\State $\sample \gets \sample \setminus \{\ensuremath{G_{R}}\}$
%	\State $\sample \gets \sample \cup \{\subgraph\}$
%\EndProcedure
\end{algorithmic}
\end{algorithm}

\subsection{Fully dynamic streams}
In this section we describe our algorithm for maintaining a uniform sample $\sample$ 
of fixed size $M$ for fully-dynamic edge streams (edge insertions and deletions).
Our algorithm relies on \emph{random pairing} (RP)~\cite{gemulla2008maintaining}, a sampling scheme that extends traditional reservoir sampling for evolving data streams, in which elements are subject to both addition and deletion operations. %, and the deleted elements are ``compensated" by future addition operations.

We first give a brief background on the RP scheme.
In RP, the uniformity of the sample is guaranteed 
by randomly pairing an inserted element with an uncompensated ``partner'' deletion, 
without necessarily keeping the identity of the partner.
At any time, there can be $0$ or more uncompensated deletions, denoted by $d$, which is equal to the difference between the cumulative number of insertions and the cumulative number of deletions.
The RP algorithm maintains ($i$) a counter $c_1$ that records the number of uncompensated deletions in which the deleted element was in the sample, ($ii$) a counter $c_2$ that records the number of uncompensated deletions in which the deleted element was not in the sample, hence, $d = c_1 + c_2$.
When $d = 0$, i.e., when there are no uncompensated deletions, 
inserted elements are processed as in standard reservoir-sampling.
When $d > 0$, the algorithm flips a coin at each inserted element and includes it in the sample 
with probability $c_1 / (c_1 + c_2)$, otherwise it excludes it from the sample (and decreases $c_1$ or $c_2$ as appropriate).

Next, we describe our adaptation of the RP scheme for fully-dynamic edge streams, 
which translate to fully-dynamic $k$-subgraph streams.
First, remember that the incremental stream translates to an incremental $k$-subgraph stream, in which connected $k$-subgraphs are only added (the first time they are created) or modified (when new induced edges arrive).

%Algorithm \ref{alg:insertionOnly} ``deterministically" compensates a deleted subgraph $H$ with the subgraph $H' = H \cup \{(u,v)\}$ that invalidates $H$ due to the addition of edge $(u,v)$ (proof of correctness provided in Claim~\ref{claim:edgeInsertionOnly}).
In the case of fully-dynamic edge streams, the $k$-connected subgraph stream is also subject to addition and deletion operations, as we explain next.
The events of interest regarding the addition of an edge have been discussed extensively in the previous section, hence we do not repeat it here. 
Consider the deletion of an edge $(u,v) \in E^{t-1}$ at time~$t$, 
and a subgraph $H \subseteq \neighborst{u,h}{t-1} \cup \neighborst{v,j}{t-1}$ in \grapht{t-1},
with $h+j=k-2$.
The effect of the edge deletion is the following: either ($i$) the vertices of $H$ remain connected, 
hence, $H$ is replaced by a new subgraph $H'$ in \grapht{t}; or ($ii$) $H$ gets disconnected, hence $H$ does not exist in \grapht{t}.
The first case corresponds to a modification of an existing connected $k$-subgraph.
As such, it does not cause an addition or deletion in the subgraph stream.

\spara{Example for k=3.}
In the case a triangle $\Delta$ in $\grapht{t-1}$
that contains an edge $(u,v)$ deleted at time $t$, 
if $\Delta \in \sample$, we modify the corresponding induced subgraph into a subgraph 
$\wedge = \Delta \setminus \{(u,v)\}$.
\hfill$\qed$
\smallskip

The second case corresponds to a deletion of a subgraph in the stream.
To handle this case, 
our sampling strategy follows the RP scheme.
In the case that a subgraph $H$ in $\grapht{t-1}$ is deleted 
due to the deletion of edge $(u,v)$ at time $t$, 
if $H \in \sample$, %  (it is in the sample), 
we increment the counter $c_1$, otherwise we increment the counter $c_2$.
In the case that a new subgraph $H'$ is formed in $\grapht{t}$ 
due to the addition of edge $(u,v)$ at time $t$, 
we include it in \sample with probability $c_1 / (c_1+c_2)$.
The approach is shown in Algorithm~\ref{alg:fullyDynamicExact}.
Next, we show that the sample $\sample$ maintained by Algorithm~\ref{alg:fullyDynamicExact} 
is uniform at any given time $t$. 

\begin{claim}\label{claim:fullyDynamicExact}
Algorithm~\ref{alg:fullyDynamicExact} ensures the uniformity of the sample $\sample$ at any time $t$. 
\end{claim}
\begin{proof}
To show that \sample is uniform, we need to consider four cases:
($i$) added edge forms a newly connected subgraph;
($ii$) deleted edge disconnects a subgraph; 
($iii$) added edge modifies an existing a subgraph; 
($iv$) deleted edge modifies an existing a subgraph.
For cases ($i$) and ($ii$), the correctness follows from RP hence we only show the correctness in cases ($iii$) and ($iv$). 
%Given that the additions in case ($i$) and deletions in case ($ii$) are handled by RP, the uniformity of the sample for these two operations directly follows from the correctness of RP.
%Now we show the correctness in cases ($iii$) and ($iv$). 
Assume the edge $(u,v) \in E^{t-1}$ is deleted (resp. added) at time $t$.
Let $H'$ denote the new subgraph due to the deletion (addition) of the edge, so that $H' = H \setminus \{(u,v)\}$ (resp. $H' = H \cup \{(u,v)\}$).
%Thus, we have, $H = H' \setminus \{(u,v)\}$ for case ($iii$) and $H = H' \cup \{(u,v)\}$  for case $(iv)$ respectively. 
Let $\sample'$ denote the sample after the invalidation of $H$ and the formation of $H'$.
Recall that $N$ remains unchanged since $H'$ replaces $H$ in \grapht{t}.
Given that the random pairing scheme guarantees uniformity of the sample at each time instance independently from the current value of $d$~\cite{gemulla2006dip}, we have $Pr\left(H \in \sample \right) = \lvert\mathcal{S} \rvert / N$. 
For the sample to be truly uniform, the probability that $H' \in \sample'$ should also be equal to $\lvert \mathcal{S} \rvert / N $ since the values of both $N$ and $\mathcal{S}$ remain unchanged as we either replace $H$ with $H'$ in $\mathcal{S}$ or we ignore $H'$ if $H \not\in \mathcal{S}$, hence $\lvert \mathcal{S}\rvert$ remains unchanged. Thus, we have,
\begin{align*}
Pr\left(H' \in \sample' \right) &= Pr\left(H \in \sample\right) \cdot Pr\left(H' \in \sample' \mid H \in \sample\right) \\ &+ Pr\left(H \not\in \sample \right) \cdot Pr\left(H' \in \sample' \mid H \not\in \sample \right) \\
&= \dfrac{\lvert \mathcal{S} \rvert}{N} \cdot 1 + \left(1 - \dfrac{\vert\mathcal{S}\rvert}{N}\right) \cdot 0 = \dfrac{\vert\mathcal{S}\rvert}{N},\\
\end{align*} 
hence uniformity is preserved.
%Given that the additions in case ($i$) and deletions in case ($ii$) are handled by RP, the uniformity of the sample for these two operations directly follows from the correctness of RP.
%Moreover, Claim~\ref{claim:edgeInsertionOnly} already covers the correctness in case ($iii$).
%Thus, we only show the correctness in case ($iv$).
%
%Assume the edge $(u,v) \in E^{t-1}$ is deleted at time $t$ and let $H'$ denote the addeded subgraph due to the deletion of the subgraph $H = H' \cup \{(u,v)\}$.
%Let $\sample'$ denote the sample after the invalidation of $H$ and the formation of $H'$.
%For the sample to be truly uniform, the probability that $H' \in \sample'$ should be equal to $M/N$, conditioned on $\sample = M < N$ (conditioning on $\sample = N < M$ is trivial since every $k$-subgraph of $G^{t}$ would then be deterministically included in $\sample'$).
%Now, given that $Pr\left(H \in \sample\right) = M/N$, we have that
%\begin{align*}
%Pr\left(H' \in \sample' \right) &= Pr\left(H \in \sample, H' \in \sample'\right)  + Pr\left(H \not\in \sample, H' \in \sample'\right) \\
%&= \dfrac{M}{N} \cdot 1 + \left(1 - \dfrac{M}{N}\right) \cdot 0 = \dfrac{M}{N},\\
%\end{align*} 
%hence uniformity is preserved.
\end{proof}

\subsection{Skip optimization}
The basic algorithm for incremental streams we described requires to process each subgraph $H \subseteq \neighborst{u,h}{t-1} \cup \neighborst{v,j}{t-1}$, for all admissible values of $h, j \ge 0$ s.t. $h+j = k-2$, to identify among them the newly created $k$-subgraphs.
All these new subgraphs are then provided as input to the standard reservoir sampling algorithm that needs to generate random numbers for each.
To reduce the cost of traversing the local neighborhood and generating a random number for each new subgraph, 
we employ \citeauthor{vitter1985random}'s acceptance-rejection algorithm that generates skip counters for reservoir sampling~\cite{vitter1985random} as follows:
let $Z_{RS}$ be the random variable that denotes the number of rejected subgraphs after the last time a subgraph was inserted to the sample $\sample$. Then, the probability that the next $z$ new subgraphs will not be accepted in $\sample$ is given by: 
\begin{align}\label{eq:pmfSkipRS}
Pr[Z_{RS}=z] = \dfrac{M}{N + z + 1} \prod_{z'=0}^{z-1} \left(1 - \dfrac{M}{N + z'+ 1} \right) .
\end{align}
Thus, rather than identifying all the new subgraphs and calling the reservoir algorithm for each, we can keep a skip counter $Z_{RS}$ that is distributed with the probability mass function given in Eq.~(\ref{eq:pmfSkipRS}), and compute its value in constant time using Vitter's acceptance-rejection algorithm for reservoir sampling~\cite{vitter1985random}. 
Then, based on the value of $Z_{RS}$ that denotes the number of new subgraph insertions we can safely skip, we can decide on the fly whether we should insert into the sample any of the new subgraphs created due to the insertion of edge $(u,v)$. 
%Then, based on the value of $Z_{RS}$ that denotes the number of subgraph insertions we can safely skip, we can decide on the fly whether and how many of the new subgraphs among the $\mathcal{W}$ newly created ones we should select into the sample. 
Given that a new $k$-subgraph $\subgraph=(\subsetv,\subsete{\subsetv})$ can be formed only when $\subsete{\subsetv} \setminus (u,v)$ is not already an induced subgraph, we can compute the \emph{exact} value of $\mathcal{W}$ as in Algorithm~\ref{alg:computeW}. The pseudocode of the optimized algorithm for incremental streams is given in Algorithm~\ref{alg:insertionSkipOptim}. 

\begin{algorithm}[tb]
\footnotesize
\caption{Compute-$\mathcal{W}$, $N^{\circ}$ of new connected $k$-subgraphs}
\label{alg:computeW}
\begin{algorithmic}[1]
\Procedure{Compute-$\mathcal{W}$}{$t, (u,v)$}
\State $\mathcal{W} \gets 0$
\For{\texttt{$h \in [0,k-2]$}}
	\State $j \gets k - 2 - h$
	\State $V_h \gets N^{t-1}_{u,h} \setminus N^{t-1}_{v,j} $
	\State $V_j \gets N^{t-1}_{v,j}  \setminus N^{t-1}_{u,h} $
	\State $x	\gets \vert\{G_S = (S, E(S)) : u \in S, |S| = h+1, S \subseteq V_h, E(S) \subseteq E(V_h)\}\vert $
	\State $y	\gets \vert\{G_S = (S, E(S)) : v \in S, |S| = j+1, S \subseteq V_j, E(S) \subseteq E(V_j)\}\vert $
%	\State $x \gets nr. of (h+1)-vertex subgraphs that contain node u in H = (V_h, E(V_h))$
%	\State $y \gets nr. of (j+1)-vertex subgraphs that contain node v in J = (V_j, E(V_j))$
	\State $\mathcal{W} \gets \mathcal{W} + x \cdot y$
\EndFor
\EndProcedure
\end{algorithmic}
\end{algorithm}

\begin{algorithm}[tb]
\footnotesize
\caption{Optimized Algorithm for Incremental Stream}
\label{alg:insertionSkipOptim}
\begin{algorithmic}[1]
\State $N \gets 0$, $\sample \gets \emptyset$, $\text{sum} \gets 0$
\State $M \gets \log{\left({\numclasses{k}}/{\delta}\right)} \cdot {(4 + \epsilon)}/{\epsilon^2}$ 
\State \Call{SkipRS}{$N,M$}:  skip function as in \cite{vitter1985random} \Comment{[\Call{SkipRS}{$N,M$} = 0 if $N < M$]} 
\Procedure{addEdge}{$t,(u,v)$}
	\For{$H \in \sample : u \in H \wedge v \in H$}
		\State $H' \gets H \cup \{(u,v)\}$ 
		\State \Call{Replace}{\sample, H, H'} \Comment{replace H with H'}
	\EndFor
	
	\State $\mathcal{W} \gets \Call{Compute-\ensuremath{\mathcal{W}}}{t, (u,v)}$ \Comment{Algorithm~\ref{alg:computeW}}
	\State $I \gets 0$ 
	\While{$\text{sum} \le \mathcal{W}$} 
		\State $I \gets I + 1$  
		\State $Z_{RS} \gets \Call{SkipRS}{N,M}$
		\State $N \gets N + Z_{RS} + 1$ 
		\State $\text{sum} \gets \text{sum} + Z_{RS} + 1$
	\EndWhile
	\State replace $I$ random elements of $\sample$ with $I$ random subgraphs drawn from $\neighborst{u,h}{t} \cup \neighborst{v,j}{t},\, \forall h \in [0,k-2], j = k - 2 - h$
	\State $\text{sum} \gets \text{sum} - \mathcal{W}$
\EndProcedure
\end{algorithmic}
\end{algorithm}

A similar optimization is also possible for fully-dynamic streams by proper adjustment of the skip counter based on the value $d = c_1 + c_2$ of uncompensated deletions . Recall that when $d = 0$, reservoir sampling is effective, hence, we can compute the value of the skip counter $Z_{RS}$ as in the case of incremental streams. When $d > 0$, the random-pairing step is effective, for which we adapt  \citeauthor{vitter1985random}'s improvements to the list-sequential sampling~\cite{vitter1984faster}. 
 
Let $Z_{RP}$ be the random variable that denotes the number of new subgraphs 
that are not accepted into the sample after the last time a subgraph was deleted (not necessarily from the sample) due to the deletion of an edge.
Assume without loss of generality that the deletion of a subgraph was followed by the creation of $d$ new subgraphs due to at least one edge insertion.
Following the fact that the new elements that random pairing includes into the sample form a uniform random sample of size $c_1$ among $d$ new elements~\cite{gemulla2008maintaining}, 
the probability that the random pairing will not accept the next $z$ new subgraphs in $\sample$ is given by: 
\begin{align}
Pr[Z_{RP}=z] = \dfrac{c_1}{d-z} \prod_{z'=0}^{z-1} \left(1 - \dfrac{c_1}{d - z'} \right) .
\end{align}
Thus, after each edge deletion, we can compute in constant time the value of skip counter $Z_{RP}$ for random pairing using acceptance-rejection algorithm for list-sequential sampling~\cite{vitter1984faster} and decide on the fly whether and how many we should insert into the sample any of the new $\mathcal{W}$ subgraphs created in the pairing step.
The algorithm to compute the exact number $\mathcal{D}$ of deleted induced subgraphs when an edge $(u,v)$ is deleted at time $t$ is similar to the computation of $\mathcal{W}$, but operates on the neighborhoods at time $t$ instead of time $t-1$ (omitted due to space constraints).
The pseudocode of the optimized algorithm is given in Algorithm~\ref{alg:fullyDynamicOptimized}.

\begin{algorithm}[tb]
\footnotesize
\caption{Optimized Algorithm for Fully-Dynamic-Edge Stream}
\label{alg:fullyDynamicOptimized}
\begin{algorithmic}[1]
\State $N \gets 0$, $\sample \gets \emptyset$, $c_1 \gets 0$, $c_2 \gets 0$
\State $M \gets \log{\left({\numclasses{k}}/{\delta}\right)} \cdot {(4 + \epsilon)}/{\epsilon^2}$ 
\State \Call{SkipRS}{$N,M$}:  skip function as in \cite{vitter1985random} \Comment{[\Call{SkipRS}{$N,M$} = 0 if $N < M$]} 
\State \Call{SkipRP}{$c_1,c_1+c_2$}:  skip function as in \cite{vitter1984faster} 
\State $\text{sum1} \gets 0, \text{sum2} \gets 0$
\Procedure{addEdge}{$t,(u,v)$}
%	\State $N^{t-1}_{uv} \gets N^{t-1}_u \cap N^{t-1}_v$
	%%%%%%%%%%%%%%%%%%%%%%%%%%%%%%%
	\For{$H \in \sample : u \in H \wedge v \in H$}
		\State $H' \gets H \cup \{(u,v)\}$ 
		\State \Call{Replace}{\sample, H, H'} \Comment{replace H with H'}
	\EndFor
	\State $\mathcal{W} \gets \Call{Compute-\ensuremath{\mathcal{W}}}{t, (u,v)}$ 
	\If{$c_1 + c_2 = 0$}
		\State $I \gets 0$  \label{line:reservoir}
		\While{$\text{sum1} \le \mathcal{W}$} 
			\State $I \gets I + 1$  
			\State $N \gets N + Z_{RS} + 1$ 
			\State $Z_{RS} \gets \Call{SkipRS}{N,M}$
			\State $\text{sum1} \gets \text{sum1} + Z_{RS} + 1$
		\EndWhile
		\State replace $I$ random elements of $\sample$ with $I$ random subgraphs drawn from $\neighborst{u,h}{t} \cup \neighborst{v,j}{t},\, \forall h \in [0,k-2], j = k - 2 - h$
		\State $\text{sum1} \gets \text{sum1} - \mathcal{W}$
	\Else
		%\State $\mathcal{W} \gets \Call{Compute-\ensuremath{\mathcal{W}}}{t, (u,v)}$ 
		\State $I \gets 0$, $sum2\gets 0$
%		\While{$\text{sum2} \le \mathcal{W}$} 
%			\State $I \gets I + 1$  
%			\State $c_1 \gets c_1 - 1$
%			\State $Z_{RP} \gets \Call{SkipRP}{c_1, c_1 + c_2}$
%			\State $\text{sum2} \gets \text{sum2} + Z_{RP} + 1$
%		\EndWhile
		\While{$\text{c1+c2} > 0 $ and $sum2 < W$} 
			\State $I \gets I + 1$
			\State $c_1 \gets c_1 - 1$
			\State $Z_{RP} \gets \Call{SkipRP}{c_1, c_1 + c_2}$
			\State $c_2 \gets c_2-Z_{RP}$  \Comment{[$c_2 = 0$ if $c_2 <0$]} 
			\State $\text{sum2} \gets \text{sum2} + Z_{RP} + 1$
		\EndWhile
		\State replace $I$ random elements of $\sample$ with $I$ random subgraphs drawn from $\neighborst{u,h}{t} \cup \neighborst{v,j}{t},\, \forall h \in [0,k-2], j = k - 2 - h$
		%\State $c_2 \gets c_2 - \mathcal{W} + I$
		%\State $\text{sum2} \gets \text{sum2} - \mathcal{W}$
		\State $\mathcal{W} \gets \mathcal{W}- sum2$
		\If{$W > 0$}
			\State Jump to line \ref{line:reservoir}
		\EndIf
	\EndIf
\EndProcedure

\Procedure{deleteEdge}{$t,(u,v)$}
	\For{$H \in \sample : u \in H \wedge v \in H$}
		\If{$H$ is still connected in \grapht{t}}
			\State $H' \gets H \setminus \{(u,v)\}$
			\State $\Call{Replace}{\sample, H, H'}$ \Comment{replace H with H'}
		\Else
		 \State $\mathcal{S} \gets \mathcal{S} \setminus H$
		 \State	$c_1 \gets c_1 + 1$
		\EndIf
	\EndFor
	\State $d \gets d + \Call{Compute-\ensuremath{\mathcal{D}}}{t, (u,v)}$
	\State $c_2 \gets d - c_1$
	\State $N \gets N - \mathcal{D}$
%	\For{$h \in [0,k-2]$}
%		\State $j \gets k -2 - h$
%		\For{$H \subseteq \neighborst{u,h}{t-1} \cup \neighborst{v,j}{t-1}$}
%			\If{$H$ is still connected in \grapht{t}}
%				\If {$H \in \sample$}
%					\State $H' \gets H \setminus \{(u,v)\}  $ 
%					\State \Call{Replace}{$\sample, H, H'$} \Comment{replace H with H'}
%				\EndIf
%			\Else
%				\If {$H \in \sample$}
%					\State $\sample \gets \sample \setminus H$
%					\State $c_1 \gets c_1 + 1$
%				\Else 
%					\State $c_2 \gets c_2 + 1$
%				\EndIf
%				\State $N \gets N-1$
%			\EndIf
%		\EndFor
%	\EndFor
\EndProcedure
%\Procedure{Replace}{$\sample, \ensuremath{G_{R}}, \subgraph$}
%	\State $\sample \gets \sample \setminus \{\ensuremath{G_{R}}\}$
%	\State $\sample \gets \sample \cup \{\subgraph\}$
%\EndProcedure
\end{algorithmic}
\end{algorithm}

\subsection{Derivation for sample size}
\label{sec:sample-size}

Now we provide a lower bound on the size of the sample $\sample$ such that 
$\approxindsubgraphs{k}{\tau}{\epsilon}{\delta}$ computed on $\sample$ provides 
an $(\epsilon, \delta)$-approximation to~$\freqindsubgraphs{k}{\tau}$. 
\begin{lemma}\label{lemma:sampSize} 
Suppose that $\vert \sample \vert = M$ satisfies 
\begin{align}\label{eq:sampleLowerBound}
M \ge \log{\left(\dfrac{\numclasses{k}}{\delta}\right)} \cdot \dfrac{(4 + \epsilon)}{\epsilon^2}
\end{align}
Then, for any isomorphism class $i \in [1,\numclasses{k}]$, $\vert \hat{p_i}  - p_i| \le \epsilon / 2$ holds with probability at least $1 - \delta / \numclasses{k}$: 
%$$\vert \hat{p_i}  - p_i| \le \epsilon / 2.$$
\end{lemma}

\begin{proof}
Let $X_i$ denote an indicator random variable that equals $1$ if a randomly sampled subgraph $\subgraph$ from $C^k$ belongs in $C^k_i$ and $0$ otherwise, $\forall i \in [1,\numclasses{k}]$.
Notice that $X_i \sim Bernoulli(p_i)$.
W.l.o.g, let $G_{j}$, $j \in [1,M]$ denote the $j$-th subgraph in $\sample$ for an arbitrary ordering of the subgraph and let $X_{i}^1, \cdots, X_{i}^M$ be iid copies of $X_i$ where each $X_{i}^j$ denotes the event $\mathbbm{1}_{[G_{j} \in C^k_i]}$.

Using the two-sided Chernoff bounds we have
\begin{align*}
Pr\left(\left\vert \sum_{j =1}^{M} X_i^j - p_i M \right\vert \ge \theta M p_i \right) \le 2\exp\left(- \dfrac{\theta^2}{2 + \theta} \cdot p_i M \right),
\end{align*} 
which implies 
\begin{align*}
Pr\left(\left\vert \hat{p}_i - p_i \right\vert \ge \theta p_i \right) \le 2\exp\left(- \dfrac{\theta^2}{2 + \theta} \cdot p_i M \right).
\end{align*} 
%Here, the two events inside the $Pr(\cdot)$ are the same events, just divided by $M$ hence the probability remains unchanged.
Now, let $\epsilon = 2p_i \theta $.
Substituting $\theta = \epsilon/(2p_i)$ we have 
\begin{eqnarray*}
Pr\left(\left\vert \hat{p}_i - p_i \right\vert \ge \epsilon / 2 \right) &\le& 2\exp\left(- \dfrac{\epsilon^2 / 4}{2p_i + \epsilon / 2} \cdot M \right). 
\end{eqnarray*} 
%\begin{eqnarray*}
%Pr\left(\left\vert \hat{p}_i - p_i \right\vert \ge \epsilon / 2 \right) &\le& 2\exp\left(- \dfrac{\epsilon^2 / (4p_i^2)}{2 + \epsilon / (2p_i)} \cdot p_i M \right) \\
%&=& 2\exp\left(- \dfrac{\epsilon^2 / 4}{2p_i + \epsilon / 2} \cdot M \right). 
%\end{eqnarray*} 
To obtain a failure probability of at most $\delta / \numclasses{k}$ for each isomorphism class $i \in [1,\numclasses{k}]$, we should have:
\begin{align*}
2\exp\left(- \dfrac{\epsilon^2 / 4}{2p_i + \epsilon / 2} \cdot M \right) \le \dfrac{\delta}{\numclasses{k}}.
\end{align*}
Rearranging the terms we obtain: 
\begin{align*}
M \ge \log{\left(\dfrac{\numclasses{k}}{\delta}\right)} \cdot \dfrac{2p_i + \epsilon/2}{(\epsilon^2 / 2)}.
\end{align*}
As we want this to hold $\forall i \in [1,\numclasses{k}]$, $M$ should satisfy:
\begin{align*}
M \ge \log{\left(\dfrac{\numclasses{k}}{\delta}\right)} \cdot \dfrac{2p_{max} + \epsilon/2}{(\epsilon^2 / 2)},
\end{align*}
where $p_{max} = \underset{i \in [1, \numclasses{k}]}{max}p_i$.
Using the worst-case $p_{max} = 1$, we obtain the following lower bound on $M$:
\begin{eqnarray*}
M &\ge& \log{\left(\dfrac{\numclasses{k}}{\delta}\right)} \cdot \dfrac{(4 + \epsilon)}{\epsilon^2}.
\end{eqnarray*}
%\begin{eqnarray*}
%M &\ge& \log{\left(\dfrac{\numclasses{k}}{\delta}\right)} \cdot \dfrac{2 + \epsilon/2}{(\epsilon^2 / 2)} \\
%&=& \log{\left(\dfrac{\numclasses{k}}{\delta}\right)} \cdot \dfrac{(4 + \epsilon)}{\epsilon^2}.
%\end{eqnarray*}
\end{proof}

\begin{theorem}
Given a uniform sample $\sample \subseteq C^k$ of size $M$ that satisfies Eq.~(\ref{eq:sampleLowerBound}), 
$\approxindsubgraphs{k}{\tau}{\epsilon}{\delta}$ provides $(\epsilon, \delta)$-approximation to  
$\freqindsubgraphs{k}{\tau}$. 
\end{theorem}

\begin{proof}
Given $M$ that satisfies Eq.~(\ref{eq:sampleLowerBound}), using union bound over all $\numclasses{k}$ estimation failure scenarios, 
we have $\vert \hat{p_i}  - p_i| \le \epsilon / 2 $, for all $i \in [1,\numclasses{k}]$, with probability at least $1-\delta$.
Then, there should be no $i \in [1, \numclasses{k}]$
with $p_i \ge \tau$, 
for which $\hat{p}_i < \tau - \epsilon / 2$. 
Hence, we ensure $\tilde{F}(\sample, \tau - \epsilon / 2) \subseteq F(C^k, \tau)$ with probability at least $1-\delta$.
Now, assume that there is a subgraph 
$\subgraph  \in C^k_i$ such that $p_i < \tau - \epsilon$.
We have that $\hat{p}_i < \tau - \epsilon / 2$, 
hence, there is no subgraph $\subgraph$ such that 
$\subgraph \not\in F(C^k, \tau)$ and $\subgraph \in \tilde{F}(\sample, \tau - \epsilon / 2)$, 
with probability at least~$1-\delta$.
\end{proof}

\section{Neighborhood Exploration}

The skip optimizations allows us to efficiently maintain the uniformity of the sample $\sample$ by eliminating the need to test the inclusion of each newly created $k$-subgraph in the local neighborhood of the inserted edge.
However, the skip optimizations require to know the number $\mathcal{W}$ of new $k$-subgraphs. 
Unfortunately, exact computation of $\mathcal{W}$ requires costly traversal of the neighborhood of the inserted edge.
Moreover, for dynamic streams, the value of the skip counter directly depends on $c_1$ and $c_2$, which require to compute the number $\mathcal{D}$ of deleted induced subgraphs after each edge deletion operation. Thus, we resort on efficient methods to approximate the values of $\mathcal{W}$ and $\mathcal{D}$. 

%\begin{algorithm}[tb]
%\footnotesize
%\caption{Compute-$\mathcal{D}$ ($k \ge 3$)}
%\label{alg:computeD}
%\begin{algorithmic}[1]
%\Procedure{Compute-$\mathcal{D}$}{$t, (u,v)$}
%\State $\mathcal{D} \gets 0$
%\For{\texttt{$h \in [0,k-2]$}}
%	\State $j \gets k - 2 - h$
%	\State $V_h \gets N^{t}_{u,h} \setminus N^{t}_{v,j} $
%	\State $V_j \gets N^{t}_{v,j}  \setminus N^{t}_{u,h} $
%	\State $x	\gets \vert\{G_S = (S, E(S)) : u \in S, |S| = h+1, S \subseteq V_h, E(S) \subseteq E(V_h)\}\vert $
%	\State $y	\gets \vert\{G_S = (S, E(S)) : v \in S, |S| = j+1, S \subseteq V_j, E(S) \subseteq E(V_j)\}\vert $
%	\State $\mathcal{D} \gets \mathcal{D} + x \cdot y$
%\EndFor
%\EndProcedure
%\end{algorithmic}
%\end{algorithm}

%\spara{Approximate computation of $\mathcal{W}$ and $\mathcal{D}$.}
To efficiently approximate the value of $\mathcal{W}$ after an edge $(u,v)$ is inserted at time $t$, we use sketches to estimate $\vert \neighborst{u,h}{t-1} \cap \neighborst{v,j}{t-1} \vert$ for all possible values of $h \in [0,k-2]$ and $j \in [0,k-2]$.
Similarly, to efficiently approximate $\mathcal{D}$  after an edge $(u,v)$ is deleted at time $t$, we use sketches to estimate $\vert \neighborst{u,h}{t} \cap \neighborst{v,j}{t} \vert$ for all possible values of $h \in [0,k-2]$ and $j \in [0,k-2]$.

Any sketching technique for set-size estimation can be used.
For our purpose, we choose to use the bottom-$k$ sketch~\citep{cohen2007summarizing} in conjunction with recently-proposed improved estimators for union and intersections of sketches~\citep{ting2016towards}.
A bottom-$k$ sketch uses a hash function $h(\cdot)$ to map elements of a universe into real numbers in $[0,1]$, and stores the $k$ minimum values in a set.
The smaller the $k$-th stored value is, the larger the size of the original set should be;
a simple estimate of the size is given by 
$
\frac{k-1}{\gamma},
$
where $\gamma$ is the largest stored hash value.

In our case, the universe of elements is the set of vertices \verticest{t} that belong to the graph at time $t$.
We build a sketch for each vertex $v \in \verticest{t}$ that summarizes \neighborst{v}{t}.
These sketches can be efficiently combined to create a sketch for the union of the neighbors of a given vertex while exploring the neighborhood via a breadth first search (BFS).

Bottom-$k$ sketches can easily be built incrementally.
When a new edge $(u,v)$ is added, we simply add the hash value of $v$ to the sketch of $u$ if it is smaller than the current maximum, and vice versa.
Alas, bottom-$k$ sketches do not directly support deletions.
However, traditionally the sketches are used in a streaming setting where memory is the main concern.
In our case, the universe of elements already resides in memory (i.e., the vertices of the graph), and our goal is to improve the speed of computation of Algorithm~\ref{alg:computeW} and its counterpart for deletion.
Therefore, we can easily store the global hash value of each vertex to be used for sketching.
Then, we can implement the sketch by using a pair of min-heap/max-heap.
The max-heap $A^+$ has bounded size and contains the hash values of the corresponding bottom-$k$ vertices.
The min-heap $A^-$ contains the hash values of the rest of the neighborhood.
Whenever an edge $(u,v)$ is deleted, if $h(v) \in A^-$ we remove the value from $A^-$ but the sketch remains unchanged; if $h(v) \in A^+$ we remove the value from $A^+$, and we also transfer the minimum value from $A^-$ to $A^+$ to maintain the fixed size of the sketch.

\subsection{Efficient implementation of \sample}
The reservoir sample \sample needs to support two main access operations efficiently: (1) Random access (to replace subgraphs in the sample, for reservoir sampling); (2) Access by vertex id (to identify modified subgraphs, as in Algorithm~\ref{alg:insertionSkipOptim}).
%\begin{squishlist}
%\item Random access (to replace subgraphs in the sample, for reservoir sampling);
%\item Access by vertex id (to identify modified subgraphs, as in Algorithm~\ref{alg:insertionSkipOptim}).
%\end{squishlist}

In order to support both operations in constant time, we resort to an array for the basic random access, supplemented by hash-based indexes for the access by vertex id.

The basic array is straightforward to implement, as the size of the sample $M$ is fixed, and the size of its element is constant $k^2$ (to store both vertices and edges).
On top of this basic array, we maintain and index $\idx : \vertices \rightarrow \left\{ S \subset \vertices \right\}$
such that $v \rightarrow S$ for all $v \in S$ and all $S \in \sample$.
That is, we have a pointer from each vertex part of a subgraph in the sample, to the set of subgraphs containing it.
Therefore, when an edge $(u,v)$ is modified at time $t$ (either added or deleted), retrieving the set of potentially affected subgraphs takes constant time.
For each potentially affected subgraph, checking whether it is actually affected also takes constant time: for a subgraph $S \in \idx(u)$ (respectively, $S \in \idx(v))$ 
we simply need to check whether $v \in S$ (respectively, $u \in S$).
If so, the subgraph needs to be updated, and so the corresponding counters for its pattern.

%%%% prev. version keep for now %%%%
%We want to perform neighborhood exploration in constant time. The skip optimization allows to perform the sampling part efficiently,
%however, we still need to compute the size of the symmetric set difference between the two neighborhoods to obtain the size of the sample space, i.e., the number of new subgraphs $\mathcal{W}$ formed by the newly arrived edge $(u,v)$
%
%\[
%\mathcal{W} = | N_{u} \bigtriangleup N_{v} | = | N_u \cup N_v \setminus N_u \cap N_v | .
%\]
%
%By the inclusion / exclusion principle,
%
%\[
%| N_u \cap N_v | = |N_u| + |N_v| - | N_u \cup N_v | ,
%\]
%and therefore, 
%\[
%\mathcal{W} = 2 \cdot | N_u \cup N_v | - |N_u| - |N_v| .
%\]
%
%We can easily estimate the size of the union by using the union of sketches for set size estimation (e.g., HyperLogLog or CountMin).
%We keep one such sketch for each vertex $v$ which contains the set of neighbors $N_v$.
%When an edge $(u,v)$ arrives, we can simply merge the two sketches to obtain a sketch for the union of the two sets.
%We can even use more sophisticated techniques for the estimation of the union and intersection of sets that have been recently developed~\citep{ting2016towards}.

\subsection{Time complexity}
Our proposed algorithms contain two components: an exploration procedure and a reservoir of samples. 
The addition of an edge $(u,v) \not\in \edgest{t-1}$ at time $t$ affects  only the subgraphs in the local neighborhoods up to $\neighborst{u,h}{t-1}$ and $\neighborst{v,j}{t-1}$, 
where $h+j=k-2$. The base algorithms, for both incremental and fully dynamic settings, iterate through the set of subgraphs in the local neighborhoods up to $\neighborst{u,h}{t-1}$ and $\neighborst{v,j}{t-1}$.
Moreover, the subgraphs are added into the reservoir in constant time, i.e., $\bigO (1)$ per subgraph, which implies that the running time of the algorithms are propotional to the expensive exploration procedure, i.e., $\bigO (\neighborst{u,h}{t-1}\cup \neighborst{v,j}{t-1})$.
The skip optimization improves the execution time by avoiding materializing and computing the expensive DFS code for many subgraphs,
but does not change its worst case upper bound. % in the local neighborhood of the edge $(u,v)$.

%Our proposed algorithms contain two components: an exploration procedure and a reservoir of samples. 
%The addition of an edge $(u,v) \not\in \edgest{t-1}$ at time $t$ affects  only the subgraphs in the local neighborhoods up to $\neighborst{u,h}{t-1}$ and $\neighborst{v,j}{t-1}$, 
%where $h+j=k-2$, i.e., all the connected $k$-subgraphs that contain $u$, $v$, 
%and $h+j$ additional nodes from their neighborhoods, for all admissible values of $h,j \ge 0$.
%The base algorithms, for both incremental and fully dynamic settings, iterate through the set of subgraphs in the local neighborhoods up to $\neighborst{u,h}{t-1}$ and $\neighborst{v,j}{t-1}$.
%Moreover, the subgraphs are added into the reservoir in constant time, i.e., $\bigO (1)$ per subgraph, which implies that the running time of the algorithms are propotional to the expensive exploration procedure, i.e., $\bigO (\neighborst{u,h}{t-1}\cup \neighborst{v,j}{t-1})$.
%The skip optimization improves the execution time by avoiding materializing and computing the expensive DFS code for many subgraphs,
%but does not change its worst case upper bound. % in the local neighborhood of the edge $(u,v)$.

% !TEX root = ../main.tex
\section{Experiments}
\label{sec:experiments}

We conduct an extensive empirical evaluation of the proposed algorithms, and provide a comparison with the existing solutions.
In particular, we answer the following interesting questions:
\begin{squishlist}
\item[\textbf{Q1:}]
What is the quality of frequent patterns for incremental streams?
\item[\textbf{Q2:}]
What is the quality of frequent patterns for dynamic streams?
\item[\textbf{Q3:}]
What is the performance in terms of average update time?
\end{squishlist}

\subsection{Experimental setup}
\spara{Datasets.}
Table~\ref{tab:summary-datasets} shows the graphs used as input in our experiments.
All datasets used are publicly available. 
Patent (PT)~\cite{hall2001nber} contains citations among US Patents from January 1963 to December 1999; 
the label of a patent is the year it was granted.
YouTube (YT)~\cite{cheng2008statistics} lists crawled videos and their related videos posted from February 2007 to July 2008. 
The label is a combination of a video's rating and length.
The streams are generated by permuting the edges in a random order.

\begin{table}[t]
\centering
\caption{Datasets used in the experiments.}
\small
\begin{tabular}{l l r r r}
\toprule
Dataset		&	Symbol	&	$| \vertices |$	&	$| \edges |$ & $| \labels |$	\\ 
\midrule
Patents 		& PT		& \num{3}M	& \num{14}M & 37 \\
Youtube		& YT		& \num{4.6}M	&	\num{43}M & 108 \\ 
\bottomrule
\end{tabular}
\label{tab:summary-datasets}
\vspace{-\baselineskip}
\end{table}

\spara{Metrics.}
We use the following metrics to evaluate the quality of all the algorithms: 
\begin{squishlist}
\item \emph{Average Relative Error (RE):}
measures how close the estimation of the frequency of the subgraph patterns compared to the ground truth. 
For the set of patterns $\freqindsubgraphs{k}{\tau}$, the average RE of the estimation is defined as $\frac{1}{\numclasses{k}} \sum_{i=1}^{\numclasses{k}} \frac{\left\vert \hat{p}_i - p_i \right\vert}{p_i}$.
\item \emph{Precision:} 
measures the fraction of frequent subgraph patterns among the ones returned by the algorithm.
\item \emph{Recall:}
measures the fraction of frequent subgraph patterns returned by the algorithm over all frequent subgraphs (as computed by the exact algorithm).
\end{squishlist} 
Additionally, we evaluate the efficiency of the algorithms by reporting the average update time. %, i.e., the average time it takes to process a single edge update.
We provide an extensive comparison of all the algorithms for $k=3$. %, while provide the results for our algorithm for $k=4$, as other algorithms do not scale well for greater values of $k$. 
We report the results of experiments averaged over $5$ runs.

\spara{Algorithms.}
%Table \ref{tab:summary-algorithm} shows the notations used for different algorithms.
We use two baselines.
Exact counting (\ec) performs exhaustive exploration of the neighborhood of the updated edge, and counts all possible subgraph patterns. 
Edge reservoir (\er) is a scheme inspired by~\citet{stefani2017triest}, which maintains a reservoir of edges during the dynamic edge updates. 
The edge reservoir is used to estimate the frequency of subgraph patterns by applying the appropriate correcting factor for the sampling probability of each pattern.
We compare these baselines with our proposed algorithms, subgraph reservoir (\sr) and its optimized version (\osr).
%We evaluate all the algorithms in both incremental and fully dynamic settings.
The size of the subgraphs reservoir is set as in Section~\ref{sec:sample-size}.
Unless otherwise specified, we fix $\epsilon=0.01$ and $\delta=0.1$.
To have a fair comparison with \er, following the evaluation of \citet{stefani2017triest}, we set the size of edge reservoir as the maximum number of edges used in the subgraph reservoir, averaged over 5 runs.
Note that \ec and \er algorithms are more competitive than any offline algorithm, e.g., \textsc{GraMi}~\citep{elseidy2014grami}, which require processing the whole graph upon any update.
\ec takes less than $2 \times 10^{-5}$ seconds to process an edge of the PT dataset, on average, while one execution of \textsc{GraMi} on the same dataset takes around $30$ seconds, which is several orders of magnitude larger, and we need to execute it once per edge.

%\begin{table}[t]
%\caption{Notation for the top-$k$ algorithms.}
%%\vspace{-2mm}
%\centering
%\small
%\begin{tabular}{l l c c  }
%\toprule
%Symbol & Reference  \\
%\midrule
%%Hashing & H \\
%Exact Counting & \ec  \\
%Edge Reservoir 		& \er	   \\
%Subgraph Reservoir & \sr  \\
%Optimized Subgraph Reservoir & \osr \\
%\bottomrule
%\end{tabular}
%\label{tab:summary-algorithm}
%\vspace{-\baselineskip}
%\end{table}

\begin{figure*}[t]
	\includegraphics[width=0.8\textwidth]{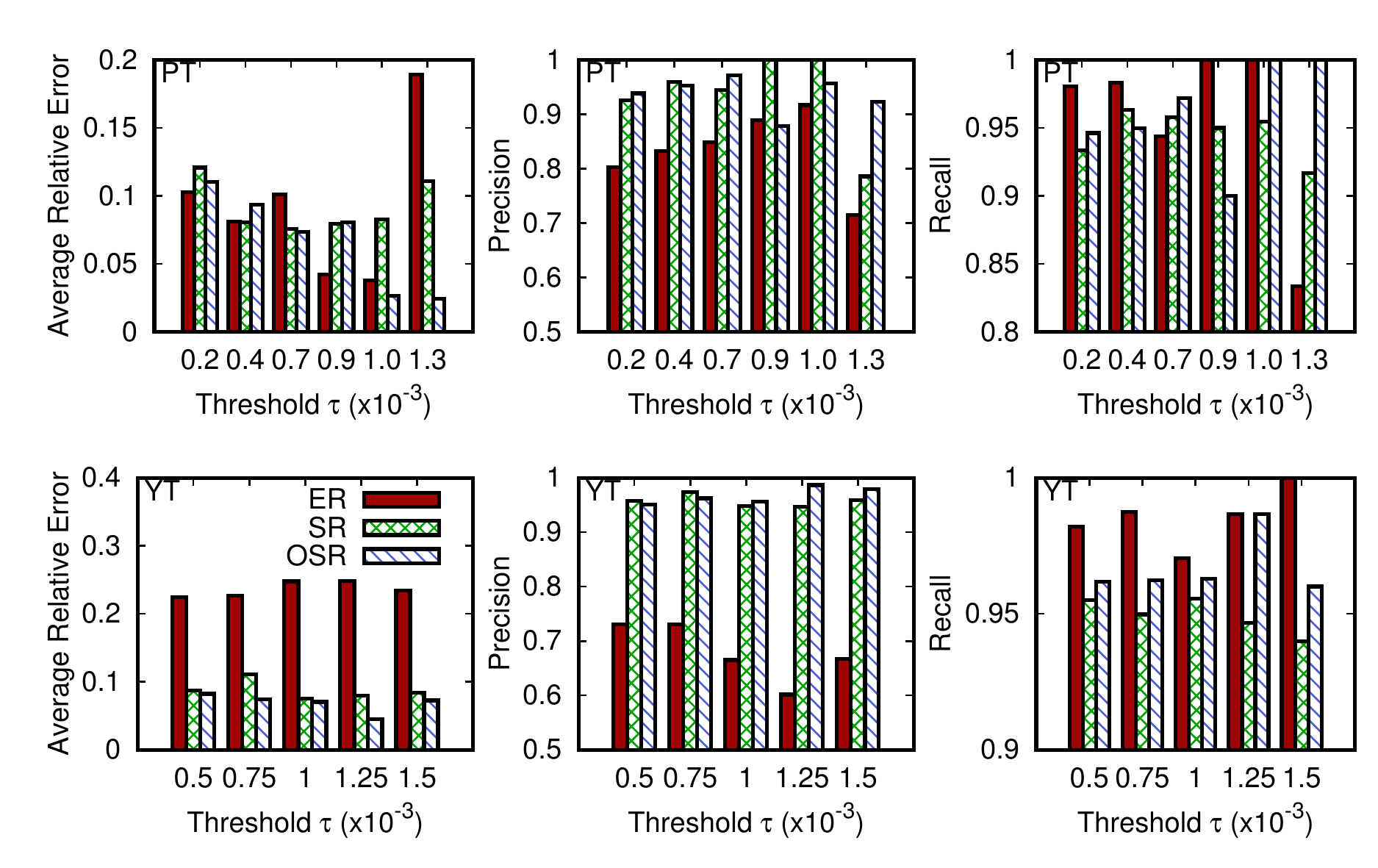}
	\caption{Relative error, precision, and recall for incremental streams on PT and YT datasets, for different values of threshold $\tau$.}
	\label{fig:incremental-quality}
\end{figure*}

\spara{Experimental environment.}
We conduct our experiments on a machine with 2 Intel Xeon Processors E5-2698 and 128GiB of memory. 
All the algorithms are implemented in Java and executed on JRE 7 running on Linux.
The source code is available online.\footnote{\url{https://github.com/anisnasir/frequent-patterns}}

\subsection{Incremental case}
We first evaluate our proposed algorithm on incremental streams.
Starting from an empty graph, we add one edge per timestamp, for both the PT and YT datasets,
and run the algorithms for several values of the frequency threshold $\threshold$.
%In particular, we show that our algorithm generates high quality results, in terms of average estimation error, precision and recall. 
%Note that the insertion-only stream does not leverage sliding windows. 
%In this experiment, we report the results for subgraph of size $k=3$.

Figure~\ref{fig:incremental-quality} shows the results.
For the PT dataset, the three algorithms behave similarly in terms of RE.
The subgraph versions offer slightly higher precision at the expense of decrease in recall.
However, for the highest frequency threshold, we see a marked deterioration of the performance of \er.
This behavior is a result of higher variance in \er due to non-uniform subgraph-sampling probabilities.
Conversely, for YT, both versions of the subgraph reservoir algorithm provide superior results in terms of average relative error.
Considering YT is the larger and more challenging dataset (in terms of number of labels), this result shows the power of subgraph sampling.
The improved estimation performance translates to much higher precision for \sr and \osr compared to \er.
The recall of all the algorithms are very similar. %, within \num{0.05} differenof each other in most cases.
Overall, the results indicate that \er generates a larger number of false positives in the result set, while \sr and \osr are able to avoid such errors while at the same time still having a low false-negative rate.

\subsection{Fully-dynamic case}
Now, we proceed to evaluate the algorithms for fully-dynamic streams. 
To produce edge deletions, we execute the algorithms in a sliding window model.
This model is of practical interest as it allows to observe recent trends in the stream.
%The number of edges in the sliding window is an input parameter.
We evaluate the algorithms for the YT dataset, and use a sliding window of size 10M.
We choose a sliding window large enough so so that the number of edges (subgraphs) do not fit in the edge (subgraph) reservoir, otherwise both algorithms are equivalent to exact counting.
We only report the results for YT dataset, as the result for the PT dataset are similar to the incremental case.

Figure~\ref{fig:fd-quality} contains the results for YT dataset.
\er obtains higher relative error compared to \sr, and poor precision and recall.
\sr is clearly the best performing algorithm in terms of accuracy, however, as we show next, it pays in terms of efficiency.
\osr has consistently better accuracy than \er, although the approximations it deploys introduce some errors.
This effect is more evident for larger frequency thresholds, where the precision drops noticeably.

\begin{figure*}[t]
	\includegraphics[width=0.8\textwidth]{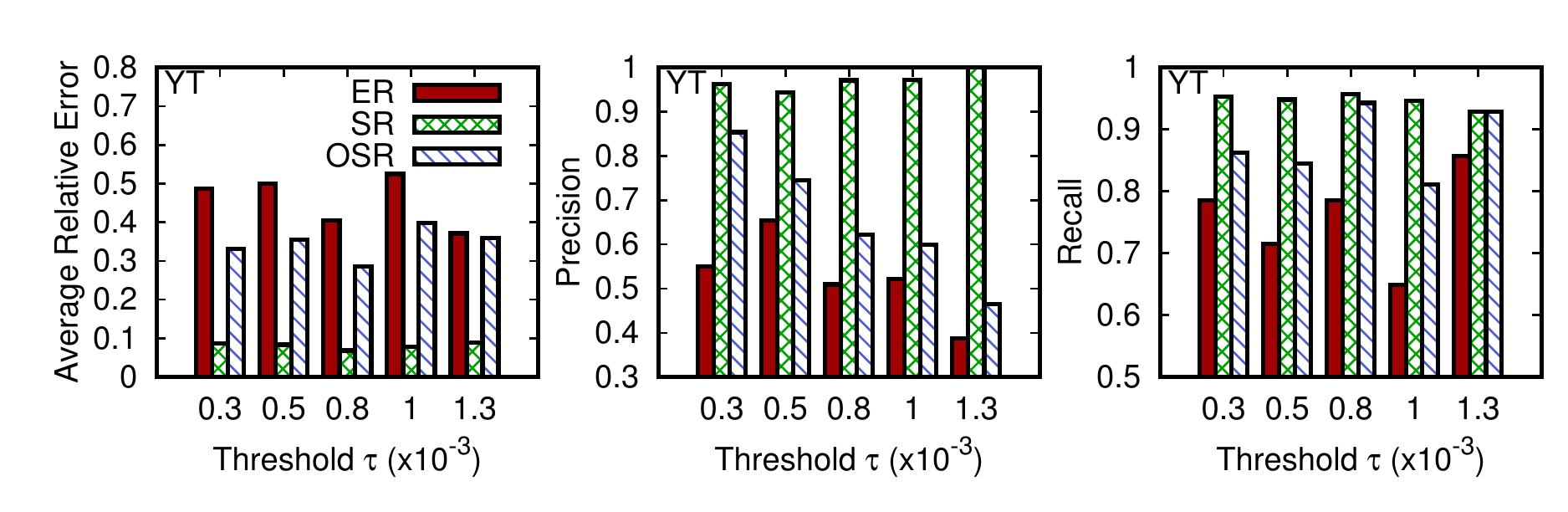}
	\caption{Relative error, precision, and recall for fully-dynamic stream on YT dataset, for different values of threshold $\tau$.}
	\label{fig:fd-quality}
\end{figure*}

\subsection{Performance}
Lastly, we evaluate the algorithms in terms of the average update time for both incremental and fully-dynamic streams on PT and YT datasets.
The size of the sliding window is 10M for the fully-dynamic streams. 
Figure~\ref{fig:performance} reports the results of the experiments which show that both \sr and \osr provide significant performance gains compared to the \ec while they are both outperformed by \er. However, given the superior accuracy of \sr and \osr compared to \er, it can be easily observed that \osr provides a good trade-off between accuracy and efficiency. 

%Furthermore, the average update times of \osr are closer to the ones of \er. We believe that this additional cost is justified compared to the gain in the quality of the \osr. 

\begin{figure}[t]
\begin{center}
	\includegraphics[width=\columnwidth]{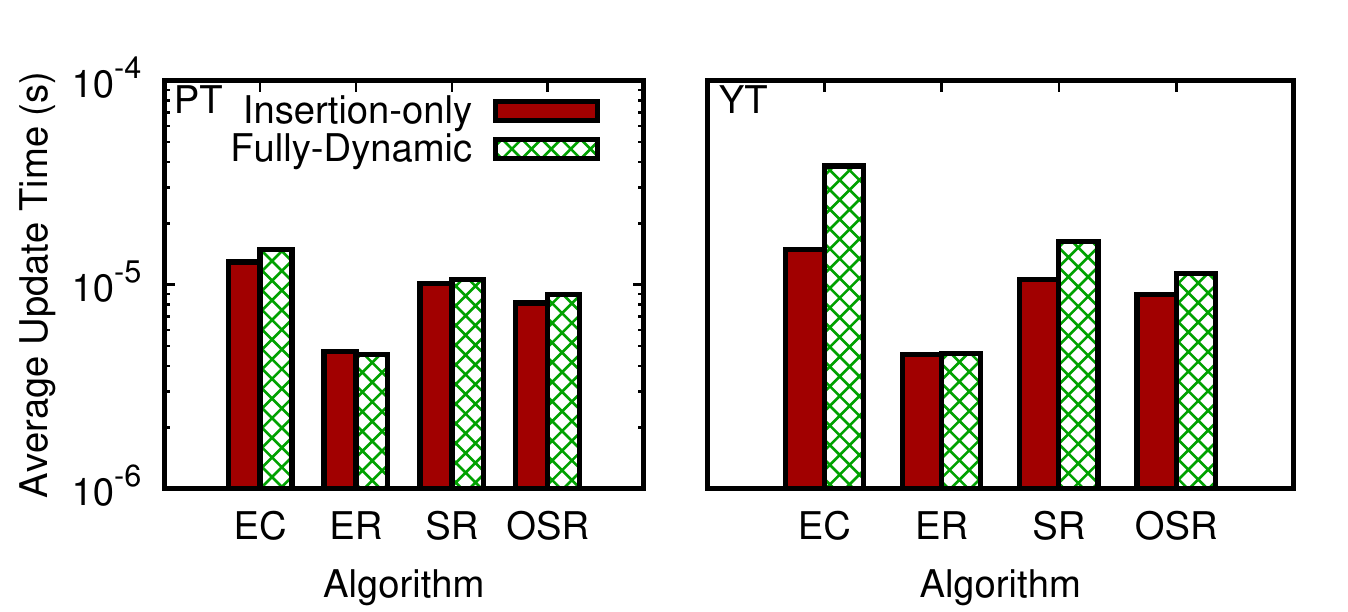}
	\caption{Average update time for incremental and fully-dynamic streams on PT and YT datasets.}
	\label{fig:performance}
\end{center}
\end{figure}

% !TEX root = ../main.tex
\section{Related Work}
\label{sec:related}

\spara{Triangle counting.} Exact and approximate triangle counting in static graphs has attracted a great deal of attention. We refer the reader to the survey by \citet{latapy2008main} for a comprehensive treatment of the topic, and include only related work on approximate triangle counting in a streaming setting. \citet{tsourakakis2011triangle} proposed triangle sparsifiers to approximate the triangle counts with a single pass of the graph, hence, the technique can also be applied to incremental streams. 
\citet{pavan2013counting} and \citet{jha2015space} proposed sampling a set of connected paths of length for approximately counting the triangles in incremental streams. %Ahmed et al.~\cite{ahmed2014graph} propose a sample-and-hold approach for insertion only streams and Bulteau et al.~\cite{bulteau2016triangle} proposed a single pass algorithm for fully dynamic streams however in both settings the count estimates could only be obtained at the end of the stream. 
\citet{lim2015mascot} proposed an algorithm based on Bernoulli sampling of edges for incremental streams, in which the edges are kept in the sample with a fixed user-defined probability. %, while the estimations are instantly updated whenever the sample includes new edges. 
Recently, \citet{stefani2017triest} proposed an algorithm for fully-dynamic streams via reservoir sampling \cite{vitter1985random} and random pairing \cite{gemulla2006dip}. 
%based on maintaing a uniform sample of edges for fully dynamic streams in which the estimations are instantly updated whenever an edge is added to or deleted from the sample.

\spara{General $k$-vertex graphlet counting.} Approximate counting of 3-, 4-, 5-vertex graphlets in static graphs has received much more attention than exact counting, which has an exponential cost. 
%Counting the exact number of graphlets has high computation cost since the number of possible k-vertex graphlets grows in $O(n^k)$. 
%Hence, approximate counting of 3-, 4-, 5-vertex graphlets in static graphs has attracted more attention than exact counting of general $k$-vertex graphlets. 
Most of the literature on approximate counting of graphlets uses random-walks to collect a uniform sample of graphlets on static graphs~\cite{bhuiyan2012guise, wang2014efficiently, han2016waddling, chen2016general}.
Alternatively, \citet{bressan2017counting} proposed a color coding based scheme for estimating $k$-vertex graphlet statistics.
Unlike the static case, approximating graphlet statistics in a streaming setting has received much less attention, and the literature is limited to incremental streams for $k > 3$. 
Wang et al.~\citep{wang2016minfer} are the first to propose an algorithm that estimates graphlet statistics from a uniform sample of edges in incremental streams. 
A recent work by~\citet{chenunified} examines approximate counting of graphlets in incremental streams for different choice of edge sampling and probabilistic counting methods. 

\spara{Transactional FSM.}
\citet{inokuchi2000apriori} introduced the problem of FSM in the transactional setting, where the goal is to mine all the frequent subgraphs on a given dataset of many, usually small, graphs.
Following \cite{inokuchi2000apriori}, a good number of algorithms for this task were provided~\citep{kuramochi2001frequent, yan2002gspan, huan2003efficient}. The transactional FSM setting is similar to the one considered by frequent-itemset mining~\cite{han2011data}, allowing to reuse many existing results, thanks to the anti-monotonicity of its support metric. %can be reused thanks to the anti-monotonicity of the support metric (i.e., count of isomorphisms).
In addition to the exact mining approaches, a line of work has studied the approximate mining of frequent subgraphs by MCMC sampling from the space of graph patterns~\cite{al2009output, saha2015fs3} with efficient pruning strategies based on anti-monotonicity of the support metric. For a comprehensive treatment, see the survey by~\citet{jiang2013survey}. 
%which do not scale to large graph datasets~\cite{chaoji2008origami}, a line of work has studied the approximate mining of frequent subgraphs by MCMC sampling from the space of graph patterns~\cite{al2009output, saha2015fs3} with efficient pruning strategies based on anti-monotonicity of the support metric. For a comprehensive treatment, see the survey by~\citet{jiang2013survey}. 

%%%%%%%%%%%
\spara{Single-Graph FSM.} 
%Unlike the transactional setting, the usual support metric obtained from the count of isomorphisms is no longer anti-monotone when mining frequent subgraphs from a given single large graph hence the literature on single-graph FSM defines several anti-monotone support metrics such as maximum independent sets (MIS)~\cite{kuramochi2005finding} or minimum image based (MNI)~\cite{bringmann2008frequent} for efficiently pruning the exhaustive search space. 
\citet{kuramochi2005finding} proposed an algorithm for exact mining of all frequent subgraphs in a given static graph that enumerates all the isomorphisms of the given graph and relies on the maximum-independent set (MIS) metric whose computation is NP-Complete. \citet{elseidy2014grami} proposed an apriori-like algorithm for exact mining of all frequent subgraphs based on the MIS metric from a given static graph. Apart from the exact mining algorithms, a line of work focused on approximate mining of frequent subgraphs in a given static graph. 
\citet{kuramochi2004grew} proposed a heuristic approach that prunes largely the search space however discovers only a small subset of frequent subgraphs without provable guarantees.
\citet{chen2007gapprox} uses an approximate version of the MIS metric, allowing approximate matches during the pruning. 
\citet{khan2010towards} propose proximity patterns, which, by relaxing the connectivity constraint of subgraphs, identify frequent patterns that cannot be found by other approaches. 

\smallskip
While the discussed work for solving FSM problem on a static graph are promising, none of them are applicable to streaming graphs.
The closest to our setting is the work by~\citet{ray2014frequent} which consider a single graph with continuous updates, however their approach is a simple heuristic applicable only to incremental streams and without provable guarantees. 
Likewise, \citet{abdelhamid2017incremental} consider an analogous setting, and propose an exact algorithm which borrows from the literature on incremental pattern mining.
The algorithm keeps track of ``fringe'' subgraph patterns, which are around the frequency threshold, and all their possible expansions/contractions (by adding/removing one edge).
While the algorithm uses clever indexing heuristics to reduce the runtime, an exact algorithm still needs to enumerate and track an exponential number of candidate subgraphs.
Finally, \citet{borgwardt2006pattern} look at the problem of finding dynamic patterns in graphs, i.e., patters over a graph time series, where persistence in time is a key element of the pattern.
By transforming the time series of a labeled edge into a binary string, the authors are able to leverage suffix trees and string-manipulation algorithms to find common substrings in the graph.
While dynamic graph patterns capture the time-series nature of the evolving graph, in our streaming scenario, only the latest instance of the graph is of interest, and the graph patterns found are comparable to the ones found for static graphs.

%\cite{teixeira2015arabesque} is one of the graph processing systems, which enables exhausting exploration for frequent subgraph mining.
%\spara{Reservoir Sampling for Dynamic Streams.}
%\citet{vitter1985random} presented a detailed analysis of the reservoir sampling scheme and discussed methods to speed up the algorithm by reducing the number of calls to the random number generator. Gemulla et al.\cite{gemulla2006dip} proposed Random Pairing as an extension of reservoir sampling for handling fully-dynamic streams. \citet{cohen2012don} generalizes and extends the Random Pairing approach to the case where the elements on the stream are
%key-value pairs, where the value may be negative (and less than −1). In our setting, where the value is not less than −1 (for an edge deletion), these generalizations do not apply hence 
%the algorithm presented by Cohen et al. \citet{cohen2012don} reduces essentially to Random Pairing.
% !TEX root = ../main.tex
\section{Conclusion}
\label{sec:conclusion}
We initiated the study of approximate frequent-subgraph mining (FSM) in both incremental and fully-dynamic streaming settings, where the edges can be arbitrarily added or removed from the graph.
For each streaming setting, we proposed algorithms that can extract a high-quality approximation of the frequent $k$-vertex subgraph patterns, for a given threshold, at any given time instance, with high probability.
Our algorithms operate by maintaining a uniform sample of $k$-vertex subgraphs at any time instance, 
for which we provide theoretical guarantees.
We also proposed several optimizations to our algorithms that allow achieving high accuracy with improved execution time.
We showed empirically that the proposed algorithms generate high-quality results compared to natural baselines.

\spara{Acknowledgements.}
Cigdem Aslay and Aristides Gionis
are supported by three Academy of Finland projects 
(286211, 313927, and 317085), and 
the EC H2020 RIA project ``SoBigData'' (654024).

\balance
\bibliographystyle{ACM-Reference-Format}
\bibliography{biblio}

%%% -*-BibTeX-*-
%%% Do NOT edit. File created by BibTeX with style
%%% ACM-Reference-Format-Journals [18-Jan-2012].

\begin{thebibliography}{00}

%%% ====================================================================
%%% NOTE TO THE USER: you can override these defaults by providing
%%% customized versions of any of these macros before the \bibliography
%%% command.  Each of them MUST provide its own final punctuation,
%%% except for \shownote{}, \showDOI{}, and \showURL{}.  The latter two
%%% do not use final punctuation, in order to avoid confusing it with
%%% the Web address.
%%%
%%% To suppress output of a particular field, define its macro to expand
%%% to an empty string, or better, \unskip, like this:
%%%
%%% \newcommand{\showDOI}[1]{\unskip}   % LaTeX syntax
%%%
%%% \def \showDOI #1{\unskip}           % plain TeX syntax
%%%
%%% ====================================================================

\ifx \showCODEN    \undefined \def \showCODEN     #1{\unskip}     \fi
\ifx \showDOI      \undefined \def \showDOI       #1{{\tt DOI:}\penalty0{#1}\ }
  \fi
\ifx \showISBNx    \undefined \def \showISBNx     #1{\unskip}     \fi
\ifx \showISBNxiii \undefined \def \showISBNxiii  #1{\unskip}     \fi
\ifx \showISSN     \undefined \def \showISSN      #1{\unskip}     \fi
\ifx \showLCCN     \undefined \def \showLCCN      #1{\unskip}     \fi
\ifx \shownote     \undefined \def \shownote      #1{#1}          \fi
\ifx \showarticletitle \undefined \def \showarticletitle #1{#1}   \fi
\ifx \showURL      \undefined \def \showURL       #1{#1}          \fi
% The following commands are used for tagged output and should be
% invisible to TeX
\providecommand\bibfield[2]{#2}
\providecommand\bibinfo[2]{#2}
\providecommand\natexlab[1]{#1}
\providecommand\showeprint[2][]{arXiv:#2}

\bibitem[\protect\citeauthoryear{Abdelhamid, Canim, Sadoghi, Bhattacharjee,
  Chang, and Kalnis}{Abdelhamid et~al\mbox{.}}{2017}]%
        {abdelhamid2017incremental}
\bibfield{author}{\bibinfo{person}{Ehab Abdelhamid}, \bibinfo{person}{Mustafa
  Canim}, \bibinfo{person}{Mohammad Sadoghi}, \bibinfo{person}{Bishwaranjan
  Bhattacharjee}, \bibinfo{person}{Yuan-Chi Chang}, {and}
  \bibinfo{person}{Panos Kalnis}.} \bibinfo{year}{2017}\natexlab{}.
\newblock \showarticletitle{Incremental Frequent Subgraph Mining on Large
  Evolving Graphs}.
\newblock \bibinfo{journal}{{\em TKDE\/}} \bibinfo{volume}{29},
  \bibinfo{number}{12} (\bibinfo{year}{2017}), \bibinfo{pages}{2710--2723}.
\newblock


\bibitem[\protect\citeauthoryear{Al~Hasan and Zaki}{Al~Hasan and Zaki}{2009}]%
        {al2009output}
\bibfield{author}{\bibinfo{person}{Mohammad Al~Hasan} {and}
  \bibinfo{person}{Mohammed~J Zaki}.} \bibinfo{year}{2009}\natexlab{}.
\newblock \showarticletitle{Output space sampling for graph patterns}.
\newblock \bibinfo{journal}{{\em PVLDB\/}} \bibinfo{volume}{2},
  \bibinfo{number}{1} (\bibinfo{year}{2009}), \bibinfo{pages}{730--741}.
\newblock


\bibitem[\protect\citeauthoryear{Bhuiyan, Rahman, and Al~Hasan}{Bhuiyan
  et~al\mbox{.}}{2012}]%
        {bhuiyan2012guise}
\bibfield{author}{\bibinfo{person}{Mansurul~A Bhuiyan},
  \bibinfo{person}{Mahmudur Rahman}, {and} \bibinfo{person}{M Al~Hasan}.}
  \bibinfo{year}{2012}\natexlab{}.
\newblock \showarticletitle{Guise: Uniform sampling of graphlets for large
  graph analysis}. In \bibinfo{booktitle}{{\em ICDM}}.
  \bibinfo{pages}{91--100}.
\newblock


\bibitem[\protect\citeauthoryear{Bifet, Holmes, Pfahringer, and
  Gavald{\`a}}{Bifet et~al\mbox{.}}{2011}]%
        {bifet2011mining}
\bibfield{author}{\bibinfo{person}{Albert Bifet}, \bibinfo{person}{Geoff
  Holmes}, \bibinfo{person}{Bernhard Pfahringer}, {and} \bibinfo{person}{Ricard
  Gavald{\`a}}.} \bibinfo{year}{2011}\natexlab{}.
\newblock \showarticletitle{Mining frequent closed graphs on evolving data
  streams}. In \bibinfo{booktitle}{{\em KDD '11}}. \bibinfo{pages}{591--599}.
\newblock


\bibitem[\protect\citeauthoryear{Borgwardt, Kriegel, and
  Wackersreuther}{Borgwardt et~al\mbox{.}}{2006}]%
        {borgwardt2006pattern}
\bibfield{author}{\bibinfo{person}{Karsten~M Borgwardt},
  \bibinfo{person}{Hans-Peter Kriegel}, {and} \bibinfo{person}{Peter
  Wackersreuther}.} \bibinfo{year}{2006}\natexlab{}.
\newblock \showarticletitle{Pattern mining in frequent dynamic subgraphs}. In
  \bibinfo{booktitle}{{\em ICDM}}. \bibinfo{pages}{818--822}.
\newblock


\bibitem[\protect\citeauthoryear{Bressan, Chierichetti, Kumar, Leucci, and
  Panconesi}{Bressan et~al\mbox{.}}{2017}]%
        {bressan2017counting}
\bibfield{author}{\bibinfo{person}{Marco Bressan}, \bibinfo{person}{Flavio
  Chierichetti}, \bibinfo{person}{Ravi Kumar}, \bibinfo{person}{Stefano
  Leucci}, {and} \bibinfo{person}{Alessandro Panconesi}.}
  \bibinfo{year}{2017}\natexlab{}.
\newblock \showarticletitle{Counting Graphlets: Space vs Time}. In
  \bibinfo{booktitle}{{\em WSDM}}. \bibinfo{pages}{557--566}.
\newblock


\bibitem[\protect\citeauthoryear{Chen, Yan, Zhu, and Han}{Chen
  et~al\mbox{.}}{2007}]%
        {chen2007gapprox}
\bibfield{author}{\bibinfo{person}{Chen Chen}, \bibinfo{person}{Xifeng Yan},
  \bibinfo{person}{Feida Zhu}, {and} \bibinfo{person}{Jiawei Han}.}
  \bibinfo{year}{2007}\natexlab{}.
\newblock \showarticletitle{gapprox: Mining frequent approximate patterns from
  a massive network}. In \bibinfo{booktitle}{{\em ICDM}}.
\newblock


\bibitem[\protect\citeauthoryear{Chen, Li, Wang, and Lui}{Chen
  et~al\mbox{.}}{2016}]%
        {chen2016general}
\bibfield{author}{\bibinfo{person}{Xiaowei Chen}, \bibinfo{person}{Yongkun Li},
  \bibinfo{person}{Pinghui Wang}, {and} \bibinfo{person}{John Lui}.}
  \bibinfo{year}{2016}\natexlab{}.
\newblock \showarticletitle{A general framework for estimating graphlet
  statistics via random walk}.
\newblock \bibinfo{journal}{{\em PVLDB\/}} \bibinfo{volume}{10},
  \bibinfo{number}{3} (\bibinfo{year}{2016}), \bibinfo{pages}{253--264}.
\newblock


\bibitem[\protect\citeauthoryear{Chen and Lui}{Chen and Lui}{2017}]%
        {chenunified}
\bibfield{author}{\bibinfo{person}{Xiaowei Chen} {and} \bibinfo{person}{John
  Lui}.} \bibinfo{year}{2017}\natexlab{}.
\newblock \showarticletitle{A unified framework to estimate global and local
  graphlet counts for streaming graphs}. In \bibinfo{booktitle}{{\em ASONAM}}.
  \bibinfo{pages}{131--138}.
\newblock


\bibitem[\protect\citeauthoryear{Cheng, Dale, and Liu}{Cheng
  et~al\mbox{.}}{2008}]%
        {cheng2008statistics}
\bibfield{author}{\bibinfo{person}{Xu Cheng}, \bibinfo{person}{Cameron Dale},
  {and} \bibinfo{person}{Jiangchuan Liu}.} \bibinfo{year}{2008}\natexlab{}.
\newblock \showarticletitle{Statistics and social network of youtube videos}.
  In \bibinfo{booktitle}{{\em IWQoS}}. \bibinfo{pages}{229--238}.
\newblock


\bibitem[\protect\citeauthoryear{Cohen and Kaplan}{Cohen and Kaplan}{2007}]%
        {cohen2007summarizing}
\bibfield{author}{\bibinfo{person}{Edith Cohen} {and} \bibinfo{person}{Haim
  Kaplan}.} \bibinfo{year}{2007}\natexlab{}.
\newblock \showarticletitle{Summarizing data using bottom-k sketches}. In
  \bibinfo{booktitle}{{\em PODC}}. \bibinfo{pages}{225--234}.
\newblock


\bibitem[\protect\citeauthoryear{Elseidy, Abdelhamid, Skiadopoulos, and
  Kalnis}{Elseidy et~al\mbox{.}}{2014}]%
        {elseidy2014grami}
\bibfield{author}{\bibinfo{person}{Mohammed Elseidy}, \bibinfo{person}{Ehab
  Abdelhamid}, \bibinfo{person}{Spiros Skiadopoulos}, {and}
  \bibinfo{person}{Panos Kalnis}.} \bibinfo{year}{2014}\natexlab{}.
\newblock \showarticletitle{Grami: Frequent subgraph and pattern mining in a
  single large graph}.
\newblock \bibinfo{journal}{{\em PVLDB\/}} \bibinfo{volume}{7},
  \bibinfo{number}{7} (\bibinfo{year}{2014}), \bibinfo{pages}{517--528}.
\newblock


\bibitem[\protect\citeauthoryear{Gemulla, Lehner, and Haas}{Gemulla
  et~al\mbox{.}}{2006}]%
        {gemulla2006dip}
\bibfield{author}{\bibinfo{person}{Rainer Gemulla}, \bibinfo{person}{Wolfgang
  Lehner}, {and} \bibinfo{person}{Peter~J Haas}.}
  \bibinfo{year}{2006}\natexlab{}.
\newblock \showarticletitle{A dip in the reservoir: Maintaining sample synopses
  of evolving datasets}. In \bibinfo{booktitle}{{\em PVLDB}}.
  \bibinfo{pages}{595--606}.
\newblock


\bibitem[\protect\citeauthoryear{Gemulla, Lehner, and Haas}{Gemulla
  et~al\mbox{.}}{2008}]%
        {gemulla2008maintaining}
\bibfield{author}{\bibinfo{person}{Rainer Gemulla}, \bibinfo{person}{Wolfgang
  Lehner}, {and} \bibinfo{person}{Peter~J Haas}.}
  \bibinfo{year}{2008}\natexlab{}.
\newblock \showarticletitle{Maintaining bounded-size sample synopses of
  evolving datasets}.
\newblock \bibinfo{journal}{{\em VLDBJ\/}} \bibinfo{volume}{17},
  \bibinfo{number}{2} (\bibinfo{year}{2008}), \bibinfo{pages}{173--201}.
\newblock


\bibitem[\protect\citeauthoryear{Hall, Jaffe, and Trajtenberg}{Hall
  et~al\mbox{.}}{2001}]%
        {hall2001nber}
\bibfield{author}{\bibinfo{person}{Bronwyn~H Hall}, \bibinfo{person}{Adam~B
  Jaffe}, {and} \bibinfo{person}{Manuel Trajtenberg}.}
  \bibinfo{year}{2001}\natexlab{}.
\newblock \bibinfo{booktitle}{{\em The NBER patent citation data file: Lessons,
  insights and methodological tools}}.
\newblock \bibinfo{type}{{T}echnical {R}eport}. \bibinfo{institution}{National
  Bureau of Economic Research}.
\newblock


\bibitem[\protect\citeauthoryear{Han and Sethu}{Han and Sethu}{2016}]%
        {han2016waddling}
\bibfield{author}{\bibinfo{person}{Guyue Han} {and} \bibinfo{person}{Harish
  Sethu}.} \bibinfo{year}{2016}\natexlab{}.
\newblock \showarticletitle{Waddling random walk: Fast and accurate sampling of
  motif statistics in large graphs}.
\newblock \bibinfo{journal}{{\em arXiv:1605.09776\/}} (\bibinfo{year}{2016}).
\newblock


\bibitem[\protect\citeauthoryear{Han, Pei, and Kamber}{Han
  et~al\mbox{.}}{2011}]%
        {han2011data}
\bibfield{author}{\bibinfo{person}{Jiawei Han}, \bibinfo{person}{Jian Pei},
  {and} \bibinfo{person}{Micheline Kamber}.} \bibinfo{year}{2011}\natexlab{}.
\newblock \bibinfo{booktitle}{{\em Data mining: concepts and techniques}}.
\newblock \bibinfo{publisher}{Elsevier}.
\newblock


\bibitem[\protect\citeauthoryear{Huan, Wang, and Prins}{Huan
  et~al\mbox{.}}{2003}]%
        {huan2003efficient}
\bibfield{author}{\bibinfo{person}{Jun Huan}, \bibinfo{person}{Wei Wang}, {and}
  \bibinfo{person}{Jan Prins}.} \bibinfo{year}{2003}\natexlab{}.
\newblock \showarticletitle{Efficient mining of frequent subgraphs in the
  presence of isomorphism}. In \bibinfo{booktitle}{{\em ICDM}}.
\newblock


\bibitem[\protect\citeauthoryear{Inokuchi, Washio, and Motoda}{Inokuchi
  et~al\mbox{.}}{2000}]%
        {inokuchi2000apriori}
\bibfield{author}{\bibinfo{person}{Akihiro Inokuchi}, \bibinfo{person}{Takashi
  Washio}, {and} \bibinfo{person}{Hiroshi Motoda}.}
  \bibinfo{year}{2000}\natexlab{}.
\newblock \showarticletitle{An apriori-based algorithm for mining frequent
  substructures from graph data}. In \bibinfo{booktitle}{{\em ECML-PKDD}}.
\newblock


\bibitem[\protect\citeauthoryear{Jha, Seshadhri, and Pinar}{Jha
  et~al\mbox{.}}{2015}]%
        {jha2015space}
\bibfield{author}{\bibinfo{person}{Madhav Jha}, \bibinfo{person}{C Seshadhri},
  {and} \bibinfo{person}{Ali Pinar}.} \bibinfo{year}{2015}\natexlab{}.
\newblock \showarticletitle{A space-efficient streaming algorithm for
  estimating transitivity and triangle counts using the birthday paradox}.
\newblock \bibinfo{journal}{{\em TKDD\/}} (\bibinfo{year}{2015}).
\newblock


\bibitem[\protect\citeauthoryear{Jiang, Coenen, and Zito}{Jiang
  et~al\mbox{.}}{2013}]%
        {jiang2013survey}
\bibfield{author}{\bibinfo{person}{Chuntao Jiang}, \bibinfo{person}{Frans
  Coenen}, {and} \bibinfo{person}{Michele Zito}.}
  \bibinfo{year}{2013}\natexlab{}.
\newblock \showarticletitle{A survey of frequent subgraph mining algorithms}.
\newblock \bibinfo{journal}{{\em The Knowledge Eng. Review\/}}
  \bibinfo{volume}{28}, \bibinfo{number}{1} (\bibinfo{year}{2013}),
  \bibinfo{pages}{75--105}.
\newblock


\bibitem[\protect\citeauthoryear{Khan, Yan, and Wu}{Khan et~al\mbox{.}}{2010}]%
        {khan2010towards}
\bibfield{author}{\bibinfo{person}{Arijit Khan}, \bibinfo{person}{Xifeng Yan},
  {and} \bibinfo{person}{Kun-Lung Wu}.} \bibinfo{year}{2010}\natexlab{}.
\newblock \showarticletitle{Towards proximity pattern mining in large graphs}.
  In \bibinfo{booktitle}{{\em SIGMOD}}. \bibinfo{pages}{867--878}.
\newblock


\bibitem[\protect\citeauthoryear{Kuramochi and Karypis}{Kuramochi and
  Karypis}{2001}]%
        {kuramochi2001frequent}
\bibfield{author}{\bibinfo{person}{Michihiro Kuramochi} {and}
  \bibinfo{person}{George Karypis}.} \bibinfo{year}{2001}\natexlab{}.
\newblock \showarticletitle{Frequent subgraph discovery}. In
  \bibinfo{booktitle}{{\em ICDM}}. \bibinfo{pages}{313--320}.
\newblock


\bibitem[\protect\citeauthoryear{Kuramochi and Karypis}{Kuramochi and
  Karypis}{2004}]%
        {kuramochi2004grew}
\bibfield{author}{\bibinfo{person}{Michihiro Kuramochi} {and}
  \bibinfo{person}{George Karypis}.} \bibinfo{year}{2004}\natexlab{}.
\newblock \showarticletitle{Grew-a scalable frequent subgraph discovery
  algorithm}. In \bibinfo{booktitle}{{\em ICDM}}.
\newblock


\bibitem[\protect\citeauthoryear{Kuramochi and Karypis}{Kuramochi and
  Karypis}{2005}]%
        {kuramochi2005finding}
\bibfield{author}{\bibinfo{person}{Michihiro Kuramochi} {and}
  \bibinfo{person}{George Karypis}.} \bibinfo{year}{2005}\natexlab{}.
\newblock \showarticletitle{Finding frequent patterns in a large sparse graph}.
\newblock \bibinfo{journal}{{\em Data mining and knowledge discovery\/}}
  \bibinfo{volume}{11}, \bibinfo{number}{3} (\bibinfo{year}{2005}),
  \bibinfo{pages}{243--271}.
\newblock


\bibitem[\protect\citeauthoryear{Latapy}{Latapy}{2008}]%
        {latapy2008main}
\bibfield{author}{\bibinfo{person}{Matthieu Latapy}.}
  \bibinfo{year}{2008}\natexlab{}.
\newblock \showarticletitle{Main-memory triangle computations for very large
  (sparse (power-law)) graphs}.
\newblock \bibinfo{journal}{{\em Theoretical Comp. Sci.\/}}
  \bibinfo{volume}{407}, \bibinfo{number}{1-3} (\bibinfo{year}{2008}),
  \bibinfo{pages}{458--473}.
\newblock


\bibitem[\protect\citeauthoryear{Lim and Kang}{Lim and Kang}{2015}]%
        {lim2015mascot}
\bibfield{author}{\bibinfo{person}{Yongsub Lim} {and} \bibinfo{person}{U
  Kang}.} \bibinfo{year}{2015}\natexlab{}.
\newblock \showarticletitle{Mascot: Memory-efficient and accurate sampling for
  counting local triangles in graph streams}. In \bibinfo{booktitle}{{\em
  KDD}}. \bibinfo{pages}{685--694}.
\newblock


\bibitem[\protect\citeauthoryear{Pavan, Tangwongsan, Tirthapura, and Wu}{Pavan
  et~al\mbox{.}}{2013}]%
        {pavan2013counting}
\bibfield{author}{\bibinfo{person}{Aduri Pavan}, \bibinfo{person}{Kanat
  Tangwongsan}, \bibinfo{person}{Srikanta Tirthapura}, {and}
  \bibinfo{person}{Kun-Lung Wu}.} \bibinfo{year}{2013}\natexlab{}.
\newblock \showarticletitle{Counting and sampling triangles from a graph
  stream}.
\newblock \bibinfo{journal}{{\em PVLDB\/}} (\bibinfo{year}{2013}).
\newblock


\bibitem[\protect\citeauthoryear{Ray, Holder, and Choudhury}{Ray
  et~al\mbox{.}}{2014}]%
        {ray2014frequent}
\bibfield{author}{\bibinfo{person}{Abhik Ray}, \bibinfo{person}{Larry Holder},
  {and} \bibinfo{person}{Sutanay Choudhury}.} \bibinfo{year}{2014}\natexlab{}.
\newblock \showarticletitle{Frequent Subgraph Discovery in Large Attributed
  Streaming Graphs}. In \bibinfo{booktitle}{{\em BigMine}}.
  \bibinfo{pages}{166--181}.
\newblock


\bibitem[\protect\citeauthoryear{Saha and Al~Hasan}{Saha and Al~Hasan}{2015}]%
        {saha2015fs3}
\bibfield{author}{\bibinfo{person}{Tanay~Kumar Saha} {and}
  \bibinfo{person}{Mohammad Al~Hasan}.} \bibinfo{year}{2015}\natexlab{}.
\newblock \showarticletitle{FS3: A sampling based method for top-k frequent
  subgraph mining}.
\newblock \bibinfo{journal}{{\em Statistical Analysis and Data Mining: The ASA
  Data Science Journal\/}} \bibinfo{volume}{8}, \bibinfo{number}{4}
  (\bibinfo{year}{2015}), \bibinfo{pages}{245--261}.
\newblock


\bibitem[\protect\citeauthoryear{Stefani, Epasto, Riondato, and Upfal}{Stefani
  et~al\mbox{.}}{2017}]%
        {stefani2017triest}
\bibfield{author}{\bibinfo{person}{Lorenzo~De Stefani},
  \bibinfo{person}{Alessandro Epasto}, \bibinfo{person}{Matteo Riondato}, {and}
  \bibinfo{person}{Eli Upfal}.} \bibinfo{year}{2017}\natexlab{}.
\newblock \showarticletitle{Tri{\`e}st: Counting local and global triangles in
  fully dynamic streams with fixed memory size}.
\newblock \bibinfo{journal}{{\em TKDD\/}} \bibinfo{volume}{11},
  \bibinfo{number}{4} (\bibinfo{year}{2017}), \bibinfo{pages}{43}.
\newblock


\bibitem[\protect\citeauthoryear{Ting}{Ting}{2016}]%
        {ting2016towards}
\bibfield{author}{\bibinfo{person}{Daniel Ting}.}
  \bibinfo{year}{2016}\natexlab{}.
\newblock \showarticletitle{Towards optimal cardinality estimation of unions
  and intersections with sketches}. In \bibinfo{booktitle}{{\em KDD '16}}.
  \bibinfo{pages}{1195--1204}.
\newblock


\bibitem[\protect\citeauthoryear{Tsourakakis, Kolountzakis, and
  Miller}{Tsourakakis et~al\mbox{.}}{2011}]%
        {tsourakakis2011triangle}
\bibfield{author}{\bibinfo{person}{Charalampos~E Tsourakakis},
  \bibinfo{person}{Mihail~N Kolountzakis}, {and} \bibinfo{person}{Gary~L
  Miller}.} \bibinfo{year}{2011}\natexlab{}.
\newblock \showarticletitle{Triangle Sparsifiers.}
\newblock \bibinfo{journal}{{\em J. Graph Algorithms Appl.\/}}
  \bibinfo{volume}{15}, \bibinfo{number}{6} (\bibinfo{year}{2011}),
  \bibinfo{pages}{703--726}.
\newblock


\bibitem[\protect\citeauthoryear{Vitter}{Vitter}{1984}]%
        {vitter1984faster}
\bibfield{author}{\bibinfo{person}{Jeffrey~Scott Vitter}.}
  \bibinfo{year}{1984}\natexlab{}.
\newblock \showarticletitle{Faster methods for random sampling}.
\newblock \bibinfo{journal}{{\em CACM\/}} \bibinfo{volume}{27},
  \bibinfo{number}{7} (\bibinfo{year}{1984}), \bibinfo{pages}{703--718}.
\newblock


\bibitem[\protect\citeauthoryear{Vitter}{Vitter}{1985}]%
        {vitter1985random}
\bibfield{author}{\bibinfo{person}{Jeffrey~S Vitter}.}
  \bibinfo{year}{1985}\natexlab{}.
\newblock \showarticletitle{Random sampling with a reservoir}.
\newblock \bibinfo{journal}{{\em TOMS\/}} \bibinfo{volume}{11},
  \bibinfo{number}{1} (\bibinfo{year}{1985}), \bibinfo{pages}{37--57}.
\newblock


\bibitem[\protect\citeauthoryear{Wackersreuther, Wackersreuther, Oswald,
  B{\"o}hm, and Borgwardt}{Wackersreuther et~al\mbox{.}}{2010}]%
        {wackersreuther2010frequent}
\bibfield{author}{\bibinfo{person}{Bianca Wackersreuther},
  \bibinfo{person}{Peter Wackersreuther}, \bibinfo{person}{Annahita Oswald},
  \bibinfo{person}{Christian B{\"o}hm}, {and} \bibinfo{person}{Karsten~M
  Borgwardt}.} \bibinfo{year}{2010}\natexlab{}.
\newblock \showarticletitle{Frequent subgraph discovery in dynamic networks}.
  In \bibinfo{booktitle}{{\em Proceedings of the Eighth Workshop on Mining and
  Learning with Graphs}}. ACM, \bibinfo{pages}{155--162}.
\newblock


\bibitem[\protect\citeauthoryear{Wang, Lui, Ribeiro, Towsley, Zhao, and
  Guan}{Wang et~al\mbox{.}}{2014}]%
        {wang2014efficiently}
\bibfield{author}{\bibinfo{person}{Pinghui Wang}, \bibinfo{person}{John Lui},
  \bibinfo{person}{Bruno Ribeiro}, \bibinfo{person}{Don Towsley},
  \bibinfo{person}{Junzhou Zhao}, {and} \bibinfo{person}{Xiaohong Guan}.}
  \bibinfo{year}{2014}\natexlab{}.
\newblock \showarticletitle{Efficiently estimating motif statistics of large
  networks}.
\newblock \bibinfo{journal}{{\em TKDD\/}} \bibinfo{volume}{9},
  \bibinfo{number}{2} (\bibinfo{year}{2014}), \bibinfo{pages}{8}.
\newblock


\bibitem[\protect\citeauthoryear{Wang, Lui, Towsley, and Zhao}{Wang
  et~al\mbox{.}}{2016}]%
        {wang2016minfer}
\bibfield{author}{\bibinfo{person}{Pinghui Wang}, \bibinfo{person}{John~CS
  Lui}, \bibinfo{person}{Don Towsley}, {and} \bibinfo{person}{Junzhou Zhao}.}
  \bibinfo{year}{2016}\natexlab{}.
\newblock \showarticletitle{Minfer: A method of inferring motif statistics from
  sampled edges}. In \bibinfo{booktitle}{{\em ICDE}}.
\newblock


\bibitem[\protect\citeauthoryear{Yan and Han}{Yan and Han}{2002}]%
        {yan2002gspan}
\bibfield{author}{\bibinfo{person}{Xifeng Yan} {and} \bibinfo{person}{Jiawei
  Han}.} \bibinfo{year}{2002}\natexlab{}.
\newblock \showarticletitle{{gSpan: Graph-Based Substructure Pattern Mining}}.
  In \bibinfo{booktitle}{{\em ICDM}}. \bibinfo{pages}{721--724}.
\newblock


\end{thebibliography}

\end{document}